\documentclass[final]{cvpr}

\usepackage{times}
\usepackage{epsfig}
\usepackage{graphicx}
\usepackage{amsmath}
\usepackage{amssymb,amsfonts}

\usepackage{appendix}
\usepackage{subfig}
\usepackage[font=bf]{caption}
\usepackage{amsthm}
\usepackage{balance}
\usepackage{xspace} 

\usepackage[font=bf]{caption}
\usepackage{algorithmic}
\usepackage[ruled,linesnumbered]{algorithm2e}

\newtheorem{theorem}{Theorem}
\newtheorem{lemma}{Lemma}

\newtheorem{corollary}{Corollary}
    
\newcommand{\argmax}{\operatornamewithlimits{argmax}}

\newcommand{\lnorm}[1]
    {\ensuremath{\left\Vert#1\right\Vert}}
\newcommand{\myparatight}[1]{\smallskip\noindent{\bf {#1}:}~}
\newcommand{\bm}[1]{\mathbf{#1}}

\usepackage[pagebackref=true,breaklinks=true,colorlinks,bookmarks=false]{hyperref}

\def\cvprPaperID{1767} 
\def\confYear{CVPR 2021}

\begin{document}

\title{PointGuard: Provably Robust 3D Point Cloud Classification}

\author{Hongbin Liu\thanks{The first two authors made equal contributions.} \qquad Jinyuan Jia\footnotemark[1] \qquad Neil Zhenqiang Gong\\
Duke University\\

{\tt\small \{hongbin.liu, jinyuan.jia, neil.gong\}@duke.edu}
}

\maketitle
\begin{abstract}
3D point cloud classification has many safety-critical applications such as autonomous driving and robotic grasping. However, several studies showed that it is vulnerable to adversarial attacks. In particular,  an attacker can make a classifier predict an incorrect label for a 3D point cloud via carefully modifying, adding, and/or deleting a small number of its points. Randomized smoothing is state-of-the-art technique to build certifiably robust 2D image classifiers. However, when applied to 3D point cloud classification, 
  randomized smoothing  can only certify robustness against adversarially {modified} points. 

In this work, we propose PointGuard, the first defense that has provable robustness guarantees against adversarially modified, added, and/or deleted points. Specifically, given a 3D point cloud and an arbitrary point cloud classifier, our PointGuard first creates multiple subsampled point clouds,  each of which contains a random subset of the points in the original point cloud; then our PointGuard predicts the label of the original point cloud as the majority vote among the labels of the subsampled point clouds predicted by the point cloud classifier. Our first major theoretical contribution is that we show PointGuard provably predicts the same label for a 3D point cloud when the number of adversarially modified, added, and/or deleted points is bounded. Our second major theoretical contribution is that we prove the tightness of our derived bound when no assumptions on the point cloud classifier are made. Moreover, we design an efficient algorithm to compute our certified robustness guarantees. We also empirically evaluate PointGuard on ModelNet40 and ScanNet benchmark datasets. 

\end{abstract}
\section{Introduction}

3D point cloud, which comprises a set of 3D points, is a crucial data structure in modelling a 3D shape or object. In recent years, we have witnessed an increasing interest in 3D point cloud classification~\cite{qi2017pointnet,li2018pointcnn,qi2017pointnet++,wang2019dynamic} because it has many  applications, such as robotic grasping~\cite{varley2017shape},  autonomous driving~\cite{chen2017multi,yue2018lidar}, etc.. However, multiple recent studies~\cite{xiang2019generating,wicker2019robustness,zheng2019pointcloud,zhou2020lg,yang2019adversarial,ma2020efficient} showed that 3D point cloud classifiers are vulnerable to adversarial attacks. In particular, given a 3D point cloud, an attacker can carefully modify, add, and/or delete a small number of points such that a 3D point cloud classifier predicts an incorrect label for it. We can categorize these attacks into four types based on the capability of an attacker: point \emph{modification}, \emph{addition}, \emph{deletion}, and \emph{perturbation} attacks. In particular, in a point modification/addition/deletion attack~\cite{xiang2019generating, wicker2019robustness,zheng2019pointcloud,yang2019adversarial}, an attacker can only modify/add/delete points in a 3D point cloud. An attacker, however, can apply one or more of the above three operations, i.e., modification, addition, and deletion, to a 3D point cloud in a point perturbation attack. Figure~\ref{fig:attack} illustrates the point modification, addition, and deletion attacks. These adversarial attacks pose severe security concerns to point cloud classification in safety-critical applications. 

\begin{figure}
    \centering
    \includegraphics[width=0.45\textwidth]{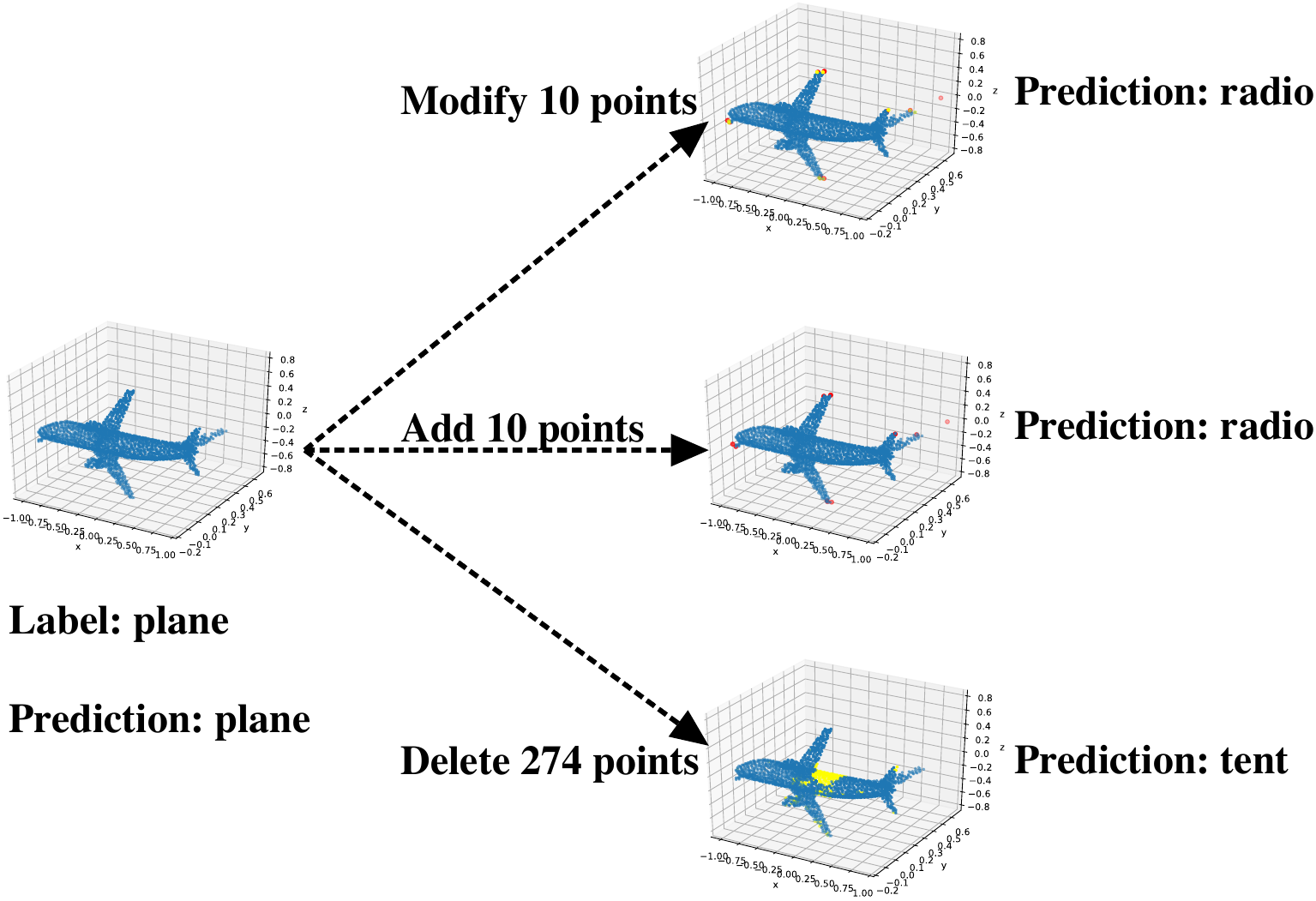}
    \caption{Top: point modification attack. Middle: point addition attack. Bottom: point deletion attack. Red points are added  and yellow points are deleted.}
    \label{fig:attack}
    \vspace{-3mm}
\end{figure}

Several \emph{empirical defenses}~\cite{liu2019extending,zhou2019dup,yang2019adversarial,dong2020self} have been proposed to mitigate the attacks. Roughly speaking, these defenses aim to detect the attacks or train more robust point cloud classifiers. For instance, Zhou et al.~\cite{zhou2019dup} proposed a defense called DUP-Net, whose key step is to detect outlier points and discard them before classifying a point cloud. These defenses, however, lack provable robustness guarantees and are often broken by more advanced attacks. For instance, Ma et al.~\cite{ma2020efficient} proposed a joint gradient based attack and showed that it can achieve high attack success rates even if DUP-Net~\cite{zhou2019dup} is deployed. 

Therefore, it is urgent to study certified defenses that have provable robustness guarantees. We say a point cloud classifier is provably robust if it certifiably predicts the same label for a point cloud when the number of modified, added, and/or deleted points is no larger than a threshold. Randomized smoothing~\cite{cohen2019certified} is state-of-the-art technique for building provably robust 2D image classifiers. For instance, via adding Gaussian noise to a 2D image, randomized smoothing    provably predicts the same label for the image when the $\ell_2$-norm of the adversarial perturbations added to the image is no larger than a threshold. Randomized smoothing can be applied to point cloud classification. For instance, we can add Gaussian noise to each point of a point cloud and randomized smoothing can predict the same label for the point cloud when the adversarial modification of its points is bounded. However,   randomized smoothing  requires the size of the input (e.g., the number of pixels in a 2D image or number of points in a 3D point cloud) remains unchanged under adversarial attacks. Therefore, randomized smoothing  is only applicable to certify robustness against the point modification attacks that do not change the size of a point cloud, leaving the other three types of attacks untouched.  

\myparatight{Our work} In this work, we propose PointGuard, the first defense that has provable robustness guarantees against point modification, addition, deletion, and perturbation attacks. Suppose we are given a 3D point cloud and an arbitrary point cloud classifier. PointGuard first creates a \emph{subsampled point cloud}, which contains a random subset of $k$ points subsampled from the original point cloud. Since the subsampled point cloud is random, its label predicted by the point cloud classifier is also random. We use $p_i$ (called \emph{label probability}) to denote the probability that the point cloud classifier predicts label $i$ for the random subsampled point cloud. Our PointGuard predicts the label that has the largest label probability for the original 3D point cloud. 

Our major theoretical contributions are twofold. First, we show that, with any point cloud classifier, our PointGuard provably predicts the same label for a point cloud when the number of modified, added, and/or deleted points is no larger than a threshold. We call the threshold \emph{certified perturbation size}. Note that the certified perturbation size may be different for different testing point clouds and point cloud classifiers. We derive the certified perturbation size via leveraging the Neyman-Pearson Lemma~\cite{neyman1933ix}. Second, we prove that, if no assumptions on the point cloud classifier are made, our derived certified perturbation size is tight, i.e.,  it is impossible to derive a certified perturbation size larger than ours. 

Our derived certified perturbation size for a point cloud is the solution to an optimization problem, which relies on the  point cloud's label probabilities. However, it is challenging to compute the exact label probabilities in practice since it requires predicting the labels for an exponential number of subsampled point clouds. In particular, computing the exact label probabilities for a point cloud requires predicting the labels for ${n \choose k}$ subsampled point clouds if the point cloud contains $n$ points. To address the challenge, we develop a Monte-Carlo algorithm to estimate the lower and upper bounds of the label probabilities with probabilistic guarantees via predicting labels for $N << {n \choose k}$ subsampled point clouds. Given the estimated label probability bounds, we solve the optimization problem to obtain the certified perturbation size.  

We empirically evaluate PointGuard on ModelNet40 and ScanNet. To demonstrate the generality of PointGuard, we consider two point cloud classifiers, i.e., PointNet~\cite{qi2017pointnet} and DGCNN~\cite{wang2019dynamic}. We adopt \emph{certified accuracy} as our evaluation metric. In particular, the certified accuracy at $r$ perturbed points   is the fraction of the testing point clouds  whose labels are correctly predicted and whose certified perturbation sizes are no smaller than $r$. Since the certified accuracy of a standard point cloud classifier is unknown, 
 we measure  a standard point cloud classifier using its \emph{empirical accuracy} under an empirical attack, i.e., we use an empirical attack to perturb the testing point clouds and use the point cloud classifier to classify them. Our experimental results show that the certified accuracy of PointGuard is substantially higher than the empirical accuracy of a standard point cloud classifier in many cases. For instance, on ModelNet40, PointNet achieves 0\% empirical accuracy while PointGuard with $k=16$ achieves 69.7\% certified accuracy when an attacker can arbitrarily modify 30 points of each testing point cloud. We also compare PointGuard with randomized smoothing for point modification attacks, as randomized smoothing is only applicable to such attacks. Our results show that PointGuard substantially outperforms randomized smoothing, e.g.,  randomized smoothing  achieves 0\% certified accuracy under the above setting.

In summary, our key contributions are as follows:
\begin{itemize}
    \item We propose PointGuard, the first 3D point cloud classification system that is provably robust against different types of adversarial attacks. 
    \item We derive the certified robustness guarantee of PointGuard and prove its tightness. Moreover, we design an algorithm to efficiently compute our certified robustness guarantee.
    \item We evaluate our PointGuard on two datasets. 
\end{itemize}

\section{Background and Related Work}
\myparatight{3D point cloud classification} A 3D point cloud is an unordered set of points sampled from the surface of a 3D object or shape. We use $T=\{O_i\ | i = 1, 2, \cdots, n \}$ to denote a 3D point cloud, where each point $O_i$ is a vector that contains the $(x, y, z)$ coordinates and possibly some other features, e.g., colours. In 3D point cloud classification, given a point cloud $T$ as input, a classifier $f$  predicts a label $y \in \{1, 2, \cdots, c\}$ for it. For instance, the label could represent the type of 3D object from which the point cloud $T$ is sampled. Formally, we have  $y = f(T)$. Many deep learning classifiers (e.g.,~\cite{qi2017pointnet,qi2017pointnet++,li2018pointcnn,wang2019dynamic}) have been proposed for 3D point cloud classification. For instance, Qi et al.~\cite{qi2017pointnet} proposed PointNet, which  can directly consume 3D point cloud. Roughly speaking, PointNet first applies input and feature transformations to the input points, and then aggregates point features by max pooling. One important characteristic of  PointNet is permutation invariant. In particular, given a  3D point cloud, the predicted label  does not rely on the order of the points in the point cloud.  

\myparatight{Adversarial attacks to 3D point cloud classification} Multiple recent works~\cite{xiang2019generating,wicker2019robustness,zheng2019pointcloud,yang2019adversarial, tsai2020robust,hamdi2020advpc,zhao2020isometry,zhou2020lg,kim2020minimal} showed that 3D point cloud classification is vulnerable to (physically feasible) adversarial attacks. Roughly speaking, given a 3D point cloud, these attacks aim to make a 3D point cloud classifier misclassify it via carefully modifying, adding, and/or deleting some points from it. Xiang et al.~\cite{xiang2019generating} proposed point perturbation and addition attacks. For instance, they showed that PointNet~\cite{qi2017pointnet} can be fooled by adding a limited number of synthesized point clusters with meaningful shapes such as balls to a point cloud. Yang et al.~\cite{yang2019adversarial} explored point modification, addition, and deletion attacks. In particular, their point modification attack is inspired by gradient-guided attack methods, which were designed to attack 2D image classification. Their point addition and deletion attacks aim to add or remove the \emph{critical} points, which can be identified by their label-dependent importance scores obtained by computing the gradient of a classifier's output with respect to the input.
Wicker et al.~\cite{wicker2019robustness} proposed a point deletion attack which also leveraged critical points. Specifically, they developed an algorithm to identify critical points in a random and iterative manner. 
Ma et al.~\cite{ma2020efficient} proposed a joint gradient based attack and showed that the proposed attack can break an empirical defense~\cite{zhou2019dup} on multiple 3D point cloud classifiers. 

\myparatight{Existing empirical defenses} Several empirical defenses~\cite{liu2019extending,zhou2019dup,yang2019adversarial,dong2020self,wu2021ifdefense,sun2021on} have been proposed to defend against adversarial attacks. Roughly speaking, these defenses aim to detect the attacks or train more robust point cloud classifiers. For instance, Zhou et al.~\cite{zhou2019dup} proposed DUP-Net, whose key idea is to detect and discard outlier points before classifying a point cloud. Dong et al.~\cite{dong2020self} designed a new self-robust 3D point recognition network. In particular, the network first extracts local features from the input point cloud and then uses a self-attention mechanism to aggregate these local features, which could ignore adversarial local features.
Liu et al.~\cite{liu2019extending} generalized adversarial training~\cite{goodfellow2015explaining} to build more robust point cloud classifiers. However, these empirical defenses lack certified robustness guarantees and are often broken by advanced adaptive attacks. 

\myparatight{Existing certified defenses} 
To the best of our knowledge, there are no certified defenses against adversarial attacks for 3D point cloud classification. We note that, however,  many certified defenses against adversarial attacks have been proposed for 2D image classification. Among these defenses, randomized smoothing~\cite{cao2017mitigating,liu2018towards,lecuyer2019certified,li2019certified,cohen2019certified,salman2019provably,pinot2019theoretical,lee2019tight, levine2020robustness,jia2020certified,wang2020certifying} is state-of-the-art because it is scalable to large-scale neural networks and applicable to arbitrary classifiers. Roughly speaking, randomized smoothing adds random noise (e.g., Gaussian noise) to an input image before classifying it; and randomized smoothing provably predicts the same label for the image when the adversarial perturbation added to the image is bounded under certain metrics, e.g., $\ell_2$ norm. We can leverage randomized smoothing to certify robustness of point cloud classification via adding random noise to each point of a point cloud. However, randomized smoothing requires the size of the input (e.g., the number of points in a point cloud) remains unaltered under adversarial attacks. As a result, randomized smoothing can only certify robustness against  point modification attacks, leaving certified defenses against the other three types of attacks untouched. 

We note that Jia et al.~\cite{jia2020intrinsic} analyzed the intrinsic certified robustness of bagging against data poisoning attacks. Both Jia et al. and our work use random sampling. 
However, our work has several key differences with Jia et al.. First, we solve a different problem from Jia et al.. More specifically, we aim to derive the certified robustness guarantees of 3D point cloud classification against adversarial attacks, while they aim to defend against data poisoning attacks. Second, we use sampling \emph{without} replacement while Jia et al. adopted sampling \emph{with} replacement, which results in significant technical differences in the derivation of the certified robustness guarantees (please refer to Supplemental Material for details). Third, we only need to train a single 3D point cloud classifier while Jia et al. requires to train multiple base classifiers.

\section{Our PointGuard}\label{method}
In this section, we first describe how to build our PointGuard from an arbitrary 3D point cloud classifier. Then, we derive the certified robustness guarantee of our PointGuard and show its tightness. Finally, we develop an algorithm to compute our certified robustness guarantee in practice.

\subsection{Building our PointGuard}
Recall that an attacker's goal is to modify, add, and/or delete a small number of points in a 3D point cloud $T$ such that it is misclassified by a point cloud classifier. Suppose we create multiple \emph{subsampled point clouds} from $T$, each of which includes $k$ points subsampled from $T$ uniformly at random without replacement. Our intuition is that, when the number of adversarially modified/added/deleted points is bounded, the majority of the subsampled point clouds do not include any adversarially modified/added/deleted points and thus the majority vote among their labels predicted by the point cloud classifier may still correctly predict the label of the original point cloud $T$. Figure~\ref{fig:pointguard} provides an example to illustrate the intuition. 

\begin{figure}[]
    \centering
    \includegraphics[width=0.47\textwidth]{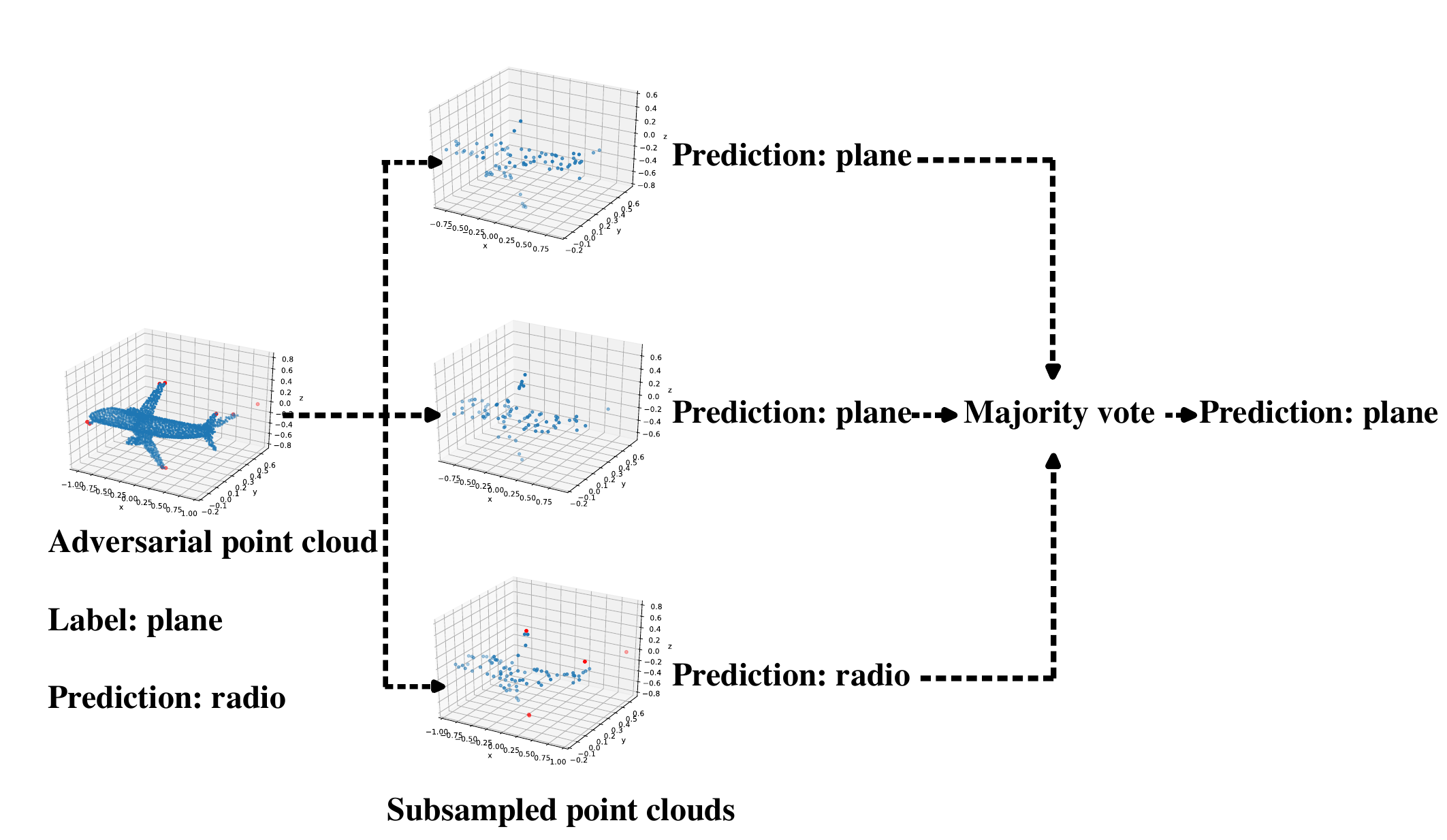}
    \caption{An example to illustrate the robustness of PointGuard. A point cloud classifier misclassifies the adversarial point cloud, where the red points are adversarially perturbed points. Three subsampled point clouds are created, and two of them (top and middle ones) do not contain adversarially perturbed points. Our PointGuard predicts the correct label for the adversarial point cloud after majority vote. }
    \label{fig:pointguard}
    \vspace{-3mm}
\end{figure} 

We design our PointGuard based on such majority-vote intuition. Next, we present a probabilistic view of the majority-vote intuition, which enables us to derive the certified robustness guarantees of PointGuard. 
Formally, we use $S_{k}(T)$ to denote a random subsampled point cloud with $k$ points from a point cloud $T$. Given an arbitrary point cloud classifier $f$, we use it to predict the label of the subsampled point cloud $S_{k}(T)$. Since the subsampled point cloud $S_{k}(T)$ is random, the predicted label $f(S_{k}(T))$ is also random. We use $p_i=\text{Pr}(f(S_{k}(T))=i)$, which we call \emph{label probability}, to denote the probability that the predicted label is $i$, where $i\in \{1, 2, \cdots, c\}$. Our PointGuard  predicts the label with the largest label probability for the point cloud $T$. For simplicity, we use $g$ to denote our PointGuard. Then, we have the following: 
\begin{align}
    g(T) = \argmax_{i \in \{1, 2, \cdots, c\}}   p_i. 
\end{align}

\subsection{Deriving the Certified Perturbation Size}
\myparatight{Certified perturbation size} In an adversarial attack, an attacker perturbs, i.e., modifies, adds, and/or deletes, some points in a point cloud $T$. We use $T^*$ to denote the perturbed point cloud. Given a point cloud $T$ and its perturbed version $T^{*}$, we define the \emph{perturbation size} $\eta(T, T^{*}) = \max \{|T|,|T^{*}|\}-|T \cap T^{*} |$, where $|\cdot|$  denotes the number of points in a point cloud. Intuitively, given $T$ and $T^{*}$, the perturbation size $\eta(T, T^{*})$ indicates the minimum number of modified, added, and/or deleted points that are required to turn $T$ into $T^{*}$. Given the point cloud $T$ and an arbitrary positive integer $r$, we define the following set: 
\begin{align}
    \Gamma(T, r) = \{T^{*}\mid  \eta(T, T^{*}) \leq r \}. 
\end{align}
Intuitively, $\Gamma(T, r) $ denotes the set of perturbed point clouds that can be obtained by perturbing at most $r$ points in $T$. 

Our goal is to find a maximum $r^*$ such that our PointGuard  provably predicts the same label for $\forall ~T^{*} \in \Gamma(T, r^*)$. Formally, we have:
\begin{align}
    r^* = \argmax_r r\ s.t.\ g(T)=g(T^{*}),\ \forall T^{*} \in \Gamma(T, r). 
\end{align}
We call $r^*$ \emph{certified perturbation size}. Note that the certified perturbation size may be different for different point clouds. 

\myparatight{Overview of our derivation} Next, we provide an overview of our proof to derive the certified perturbation size of our PointGuard for a point cloud $T$. The detailed proof is shown in the Supplemental Material. 
Our derivation is inspired by previous work~\cite{jia2019certified,jia2020intrinsic}. In particular, the key idea is to divide the space into different regions based on the Neyman-Pearson Lemma~\cite{neyman1933ix}. However, due to the difference in sampling methods, our space divisions are significantly different from previous work~\cite{jia2020intrinsic}.
For simplicity, we define random variables $\mathbf{W} = S_{k}(T)$ and $\mathbf{Z} = S_{k}(T^*)$, which represent the random subsampled point clouds from $T$ and $T^{*}$, respectively. Given these two random variables, we denote $p_i =  \text{Pr}(f(\mathbf{W})=i)$ and $p_i^{*} = \text{Pr}(f(\mathbf{Z})=i)$, where $i \in \{1, 2, \cdots, c\}$. Moreover, we denote $y =g(T)= \argmax_{i\in\{1, 2, \cdots, c\}} p_i $. Our goal is to find the maximum $r^*$ such that $y =g(T^*) = \argmax_{i\in\{1, 2, \cdots, c\}} p_i^{*}$ (i.e., $p_{y}^{*} > \max_{i \neq y} p_{i}^{*}$) for $\forall T^{*} \in \Gamma(T, r^*)$. 

The major challenge in deriving the certified perturbation size $r^{*}$ is to compute $p_i^{*}$. Specifically, the challenge stems from the complexity of the point cloud classifier $f$ and predicting labels for the ${t \choose k}$ subsampled point clouds $S_k(T^*)$, where $t$ is the number of points in $T^*$. To overcome the challenge, we propose to derive a lower bound of $p_{y}^{*}$ and an upper bound of $\max_{i \neq y} p_{i}^{*}$. Moreover, based on the Neyman-Pearson Lemma~\cite{neyman1933ix}, we derive the lower/upper bounds as the probabilities that the random variable $\mathbf{Z}$ is in certain regions of its domain space, which can be efficiently computed for any given $r$.  Then, we find the certified perturbation size $r^{*}$ as the maximum $r$ such that the lower bound of $p_{y}^{*}$ is larger than the upper bound of $\max_{i \neq y} p_{i}^{*}$. 

Next, we discuss how we derive the lower/upper bounds. Suppose we have a lower bound $\underline{p_y}$ of $p_y$  and an upper bound $\overline{p}_e$ of the second largest label probability $p_e$ for the original point cloud $T$. Formally, we have the following: 
\begin{align}
\label{probability_bound_main}
    p_y  \geq \underline{p_y} \geq  \overline{p}_e  \geq p_e = \max_{i \neq y}  p_i,  
\end{align}
where $\underline{p}$ and $\overline{p}$ denote the lower and upper bounds of $p$, respectively. Note that $\underline{p_y}$ and $\overline{p}_e$ can be estimated using the unperturbed point cloud $T$. In Section~\ref{estimation_in_practice}, we propose an algorithm to estimate them. Based on the fact that $p_y$ and $p_i$ $(i \neq y)$ should be integer multiples of $1/{n \choose k}$, where $n$ is the number of points in $T$, we have the following:
\begin{align}
\label{probability_condition_advance}
& \underline{p^{\prime}_y} \triangleq \frac{\lceil \underline{p_y}\cdot{n \choose k} \rceil}{{n \choose k}} \leq \text{Pr}(f(\mathbf{W})= y), \\
\label{probability_condition_advance_2}
& \overline{p}^{\prime}_i \triangleq \frac{\lfloor \overline{p}_i \cdot{n \choose k} \rfloor}{{n \choose k}}  \geq \text{Pr}(f(\mathbf{W})=i), \forall i \neq y.
\end{align}
Given these probability bounds $\underline{p^{\prime}_y}$ and $\overline{p}^{\prime}_i$ $(i \neq y)$, we derive a lower bound of $p_y^*$ and an upper bound of $p_i^{*}$ $(i \neq y)$ via the Neyman-Pearson Lemma~\cite{neyman1933ix}. We use $\Phi$ to denote the joint space of $\mathbf{W}$ and $\mathbf{Z}$, where each element in the space is a 3D point cloud with $k$ points subsampled from $T$ or $T^*$. We denote by $E$ the set of intersection points between $T$ and $T^*$, i.e., $E = T \cap T^*$. Then, we can divide $\Phi$ into three disjoint regions: $\Delta_T$, $\Delta_{E}$, and $ \Delta_{T^*}$. In particular, $\Delta_E$ consists of the subsampled point clouds that can be obtained by subsampling $k$ points from $E$; and $\Delta_T$ (or $\Delta_{T^{*}}$) consists of the subsampled point clouds that are subsampled from $T$ (or $T^*$) but do not belong to $\Delta_E$. 

We assume $\underline{p_y^{\prime}} > \text{Pr}(\mathbf{W} \in \Delta_{T})$. Note that we can make this assumption because our goal is to find a sufficient condition. Then, we can find a region $\Delta_y \subseteq \Delta_{E}$ such that $\text{Pr}(\mathbf{W} \in \Delta_y) = \underline{p_y^{\prime}} - \text{Pr}(\mathbf{W} \in \Delta_{T})$. We can find the region because  $\underline{p_y^{\prime}} - \text{Pr}(\mathbf{W} \in \Delta_{T})$ is an integer multiple of $1/{n \choose k}$. Similarly, we can assume $\overline{p}_i^{\prime} < \text{Pr}(\mathbf{W} \in \Delta_{E})$ since our goal is to find a sufficient condition. Then, for each $i \neq y$, we can find a region $\Delta_i \subseteq \Delta_E$ such that $\text{Pr}(\mathbf{W} \in \Delta_i) = \overline{p}_{i}^{\prime}$ based on the fact that $\overline{p}_{i}^{\prime}$ is an integer multiple of $1/{n \choose k}$. Finally, we derive the following bounds based on the Neyman-Pearson Lemma~\cite{neyman1933ix}: 
\begin{align}
& p_y^* \geq   \underline{p_y^*} = \text{Pr}(\mathbf{Z} \in \Delta_y), \\
&  p_i^* \leq \overline{p}_i^* = \text{Pr}(\mathbf{Z} \in \Delta_i \cup \Delta_{T^*}), \forall i \neq y,
\end{align}
where $\text{Pr}(\mathbf{Z} \in \Delta_y)$ and $\text{Pr}(\mathbf{Z} \in \Delta_i \cup \Delta_{T^*})$ represent the probabilities that $\mathbf{Z}$ is in the corresponding regions, which can be efficiently computed via the probability mass function of $\mathbf{Z}$. Then, our certified perturbation size $r^*$ is the maximum $r$ such that $\underline{p_y^*} > \max_{i \neq y} \overline{p}_i^*$. 

Formally, we have the following theorem: 
\begin{theorem}[Certified Perturbation Size]
\label{certified_radius_pointcloud}
Suppose we have an arbitrary point cloud classifier $f$, a 3D point cloud  $T$, and a subsampling size $k$. $y$, $e$, $\underline{p_y}\in[0,1]$, and $\overline{p}_e \in[0,1]$  satisfy  Equation~(\ref{probability_bound_main}). Then,  our PointGuard $g$ guarantees that $g(T^*)= y$, $\forall T^*\in {\Gamma}(T,r^*)$, where $r^{*}$ is the solution to the following optimization problem: 
\begin{align}
r^{*} & =  \argmax_{r} r \nonumber \\
\label{certified_condition_pointcloud}
 \text{s.t.} & \max_{n-r \leq t \leq n+r}\frac{{t \choose k}}{{n \choose k}} - 2 \cdot \frac{{\max(n,t)-r \choose k}}{{n \choose k}} 
 + 1 - \underline{p_y^{\prime}}+\overline{p}_e^{\prime}  < 0, 
\end{align}
where $\underline{p_y^{\prime}} $ and $\overline{p}_e^{\prime}$ are respectively defined in Equation~(\ref{probability_condition_advance}) and~(\ref{probability_condition_advance_2}), $n$ is the number of points in $T$, and $t$ is the number of points in $T^*$ which ranges from $n-r$ to $n+r$ when the perturbation size is $r$.
\end{theorem}
\begin{proof}
See Section~\ref{proof_of_certified_radius_bagging} in Supplementary Material. 
\end{proof}

Note that our Theorem~\ref{certified_radius_pointcloud} can be applied to any of the four types of adversarial attacks to point cloud classification, i.e., point perturbation, modification, addition, and deletion attacks. Moreover, for point modification, addition, and deletion attacks, we can further simplify the constraint in Equation~(\ref{certified_condition_pointcloud}) as there is a simple relationship between $t$, $n$, and $r$. Specifically, we have the following corollaries. 
\begin{corollary}[Point Modification Attacks]
\label{modification}
Suppose an attacker only modifies existing points in a 3D point cloud, i.e., we have $t=n$. Then, the constraint in Equation~(\ref{certified_condition_pointcloud})  reduces to $1 - \frac{{n-r \choose k}}{{n \choose k}} - \frac{\underline{p_y^{\prime}}-\overline{p}_e^{\prime}}{2} < 0$.
\end{corollary}

\begin{corollary}[Point Addition Attacks]
\label{addition}
Suppose an attacker  only adds new points to a 3D point cloud, i.e., we have $t=n+r$. Then, the constraint in Equation~(\ref{certified_condition_pointcloud})   reduces to $\frac{{n+r \choose k}}{{n \choose k}} - 1 -  \underline{p_y^{\prime}}+\overline{p}_e^{\prime}  < 0$.
\end{corollary}

\begin{corollary}[Point Deletion Attacks]
\label{deletion}
Suppose an attacker  only deletes existing points from a 3D point cloud, i.e., we have $t=n-r$. Then, the constraint in Equation~(\ref{certified_condition_pointcloud}) reduces to $- \frac{{n-r \choose k}}{{n \choose k}} + 1 - \underline{p_y^{\prime}}+\overline{p}_e^{\prime} < 0$.
\end{corollary}

Next, we  show that our derived certified perturbation size is tight, i.e., it is impossible to derive a  certified perturbation size larger than ours if no assumptions  on the point cloud classifier $f$ are made. 
\begin{theorem}[Tightness of certified perturbation size]
\label{tightness}
Suppose we have $\underline{p^{\prime}_y}+\overline{p}^{\prime}_e \leq 1$ and $\underline{p^{\prime}_y}+\sum_{i \neq y}\overline{p}^{\prime}_i \geq 1$. Then, for $\forall r > r^*$, there exist a point cloud classifier $f^*$ which satisfies Equation~(\ref{probability_bound_main}) and an adversarial point cloud $T^*$ such that $g(T^*) \neq y$ or there exist ties. 
\end{theorem}
\begin{proof}
See Section~\ref{proof_of_tightness} in Supplementary Material. 
\end{proof}

\begin{algorithm}[tb]
   \caption{ \textsc{Prediction \& Certification}}
   \label{alg:certify}
\begin{algorithmic}
   \STATE {\bfseries Input:} $f$, $T$, $k$, $N$, and $\alpha$. 
   \STATE {\bfseries Output:} Predicted label and certified perturbation size. \\
   $M_1, M_2,\cdots,M_N \gets  \textsc{RandomSubsample}(T,k)$ \\
   \STATE counts$[i] \gets \sum_{j=1}^{N}\mathbb{I}(f(M_j)=i), i = 1,2,\cdots,c. $ \\
   \STATE ${y}, {e} \gets \text{top two indices in counts}$ \\
   \STATE $\underline{p_{{y}}}, \overline{p}_{{e}} \gets \textsc{ProbBoundEstimation}(\text{counts},\alpha)$ \\
   \IF{$\underline{p_{{y}}}> \overline{p}_{{e}}$}
    \STATE $r^* = \textsc{BinarySearch}(|T|, k, \underline{p_{{y}}}, \overline{p}_{{e}})$
  \ELSE
  \STATE ${y},r^* \gets \text{ABSTAIN},\text{ABSTAIN}$ 
  \ENDIF
  \STATE \textbf{return} ${y},r^*$
\end{algorithmic}
\end{algorithm}

\subsection{Computing the Certified Perturbation Size}
\label{estimation_in_practice}
Given an arbitrary point cloud classifier $f$ and a 3D point cloud $T$,  computing the certified perturbation size $r^*$ requires solving the optimization problem in Equation~(\ref{certified_condition_pointcloud}), which involves the label probability lower bound $\underline{p_y}$ and upper bound $\overline{p}_e$. We develop a Monte-Carlo algorithm to estimate these label probability bounds, with which we solve the optimization problem efficiently via binary search. 

\myparatight{Estimating label probability lower and upper bounds} We first create $N$ random subsampled point clouds from $T$, each of which contains $k$ points. For simplicity, we use $M_1, M_2, \cdots, M_{N}$ to denote them. Then, we  use the point cloud classifier $f$ to predict a label for each subsampled point cloud $M_j$, where $j=1, 2, \cdots, N$. We use $N_i$ to denote the number of subsampled point clouds whose predicted labels are  $i$, i.e., $N_i = \sum_{j=1}^{N}\mathbb{I}(f(M_j)=i)$, where $\mathbb{I}$ is an indicator function. We predict the label with the largest  $N_i$ as the label of the original point cloud $T$, i.e., $y=g(T)=\argmax_{i\in \{1,2,\cdots,c\}} N_i$. 
Moreover, based on the definition of the label probability $p_i$, we know $N_i$ follows a binomial distribution with parameters $N$ and $p_i$, i.e., $N_i \sim \text{Binomial}(N, p_i)$. Therefore, we can use SimuEM~\cite{jia2019certified}, which is based on Clopper-Pearson~\cite{clopper1934use} method, to estimate a lower or upper bound of $p_i$ using $N_i$ and $N$. In particular, we have: 
\begin{align}
\label{prob_bound_est_1}
 &   \underline{p_y} = \text{Beta}(\frac{\alpha}{c}; N_y, N - N_y + 1), \\
 \label{prob_bound_est_2}
      &  \overline{p}_i = \text{Beta}(1-\frac{\alpha}{c}; N_i, N - N_i + 1), \forall i \neq y,
\end{align}
where $1-\alpha$ is the confidence level for simultaneously estimating the $c$ label probability bounds and $\text{Beta}(\tau ; \mu,  \nu)$ is the $\tau$th quantile of the Beta distribution with shape parameters $\mu$ and $\nu$. Both $\max_{i\neq y}\overline{p}_i$ and $1 - \underline{p_y}$ are upper bounds of $\overline{p}_e$. We use the smaller one as $\overline{p}_e$, i.e., we have $\overline{p}_e = \min( \max_{i\neq y}\overline{p}_i, 1 - \underline{p_y} )$, which gives a tighter $\overline{p}_e$. 

\myparatight{Solving the optimization problem} Given the estimated label probability bounds, we can use binary search to solve the optimization problem in Equation~(\ref{certified_condition_pointcloud}) to find the certified perturbation size $r^*$. 

\myparatight{Complete algorithm}
Algorithm~\ref{alg:certify} shows the complete process of our prediction and certification for a 3D point cloud $T$, which outputs our PointGuard's predicted label $y$ and certified perturbation size $r^*$ for $T$. The function \textsc{RandomSubsample} creates $N$ subsampled point clouds from $T$. The function \textsc{ProbBoundEstimation} estimates the label probability lower and upper bounds $\underline{p_{{y}}}$ and $\overline{p}_{{e}}$ with confidence level $1-\alpha$ based on Equation~(\ref{prob_bound_est_1}) and~(\ref{prob_bound_est_2}). The function \textsc{BinarySearch} solves the optimization problem in Equation~(\ref{certified_condition_pointcloud})  to obtain $r^*$ using binary search.  

\subsection{Training the Classifier with Subsampling}
Our PointGuard is built upon a point cloud classifier $f$. In particular, in the standard training process, $f$ is trained on the original point clouds. Our PointGuard uses $f$ to predict the labels for the subsampled point clouds. 
Since the subsampled point clouds have a different distribution from the original point clouds,  $f$ has a low  accuracy on the subsampled point clouds. As a result, PointGuard has suboptimal robustness. To address the issue, 
 we propose to train the point cloud classifier $f$ on subsampled point clouds instead of the original point clouds. Specifically, given a batch of  point clouds from the training dataset, we first create a subsampled point cloud for each  point cloud in the batch, and then we use the batch of subsampled point clouds to update $f$. Our experimental results show that  training with subsampling significantly improves the robustness of PointGuard. 
\section{Experiments}
\subsection{Experimental Setup}

\myparatight{Datasets and models}
We evaluate PointGuard on ModelNet40~\cite{wu20153d} and  ScanNet~\cite{dai2017scannet} datasets. In particular, the ModelNet40 dataset contains 12,311 3D CAD models, each of which is a point cloud comprising 2,048 three dimensional points and  belongs to one of the 40 common object categories. The dataset is splitted into 9,843 training point clouds and 2,468 testing point clouds. ScanNet is an RGB-D video dataset which contains 2.5M views from 1,513 scenes. Following Li et al.~\cite{li2018pointcnn}, we extract 6,263 training point clouds and 2,061 testing point clouds from ScanNet, each of which belongs to one of the 16 object categories and has 2,048 six dimensional points. We adopt PointNet~\cite{qi2017pointnet} and DGCNN~\cite{wang2019dynamic} as the point cloud classifiers for ModelNet40 and ScanNet, respectively. We use publicly available implementations for both PointNet\footnote{https://github.com/charlesq34/pointnet} and DGCNN\footnote{https://github.com/WangYueFt/dgcnn}. 

\begin{figure*}[!t]
 \vspace{-2mm}
\centering
\subfloat[Point perturbation attacks.]{\includegraphics[width=0.24\textwidth]{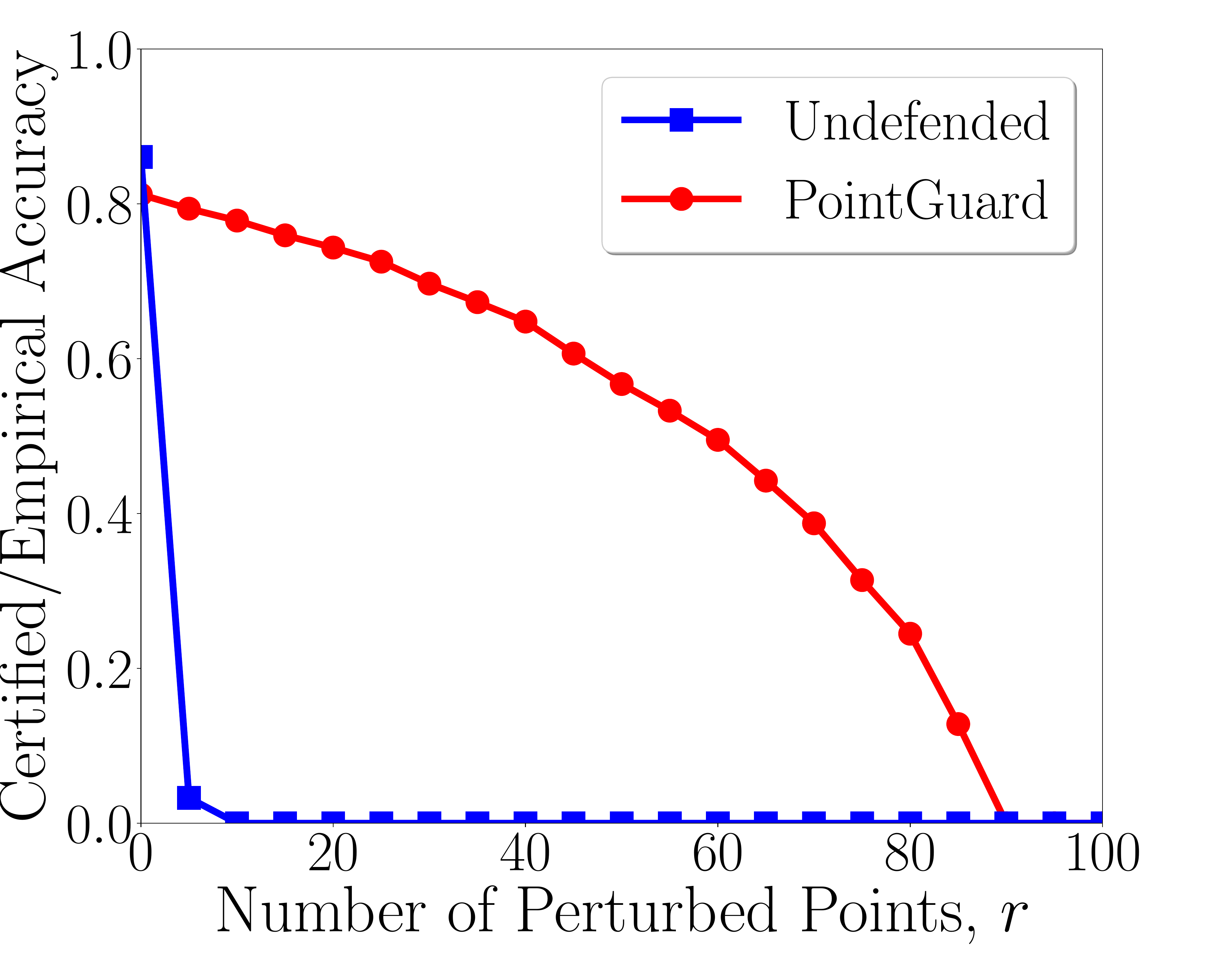}}
\subfloat[Point modification attacks.]{\includegraphics[width=0.24\textwidth]{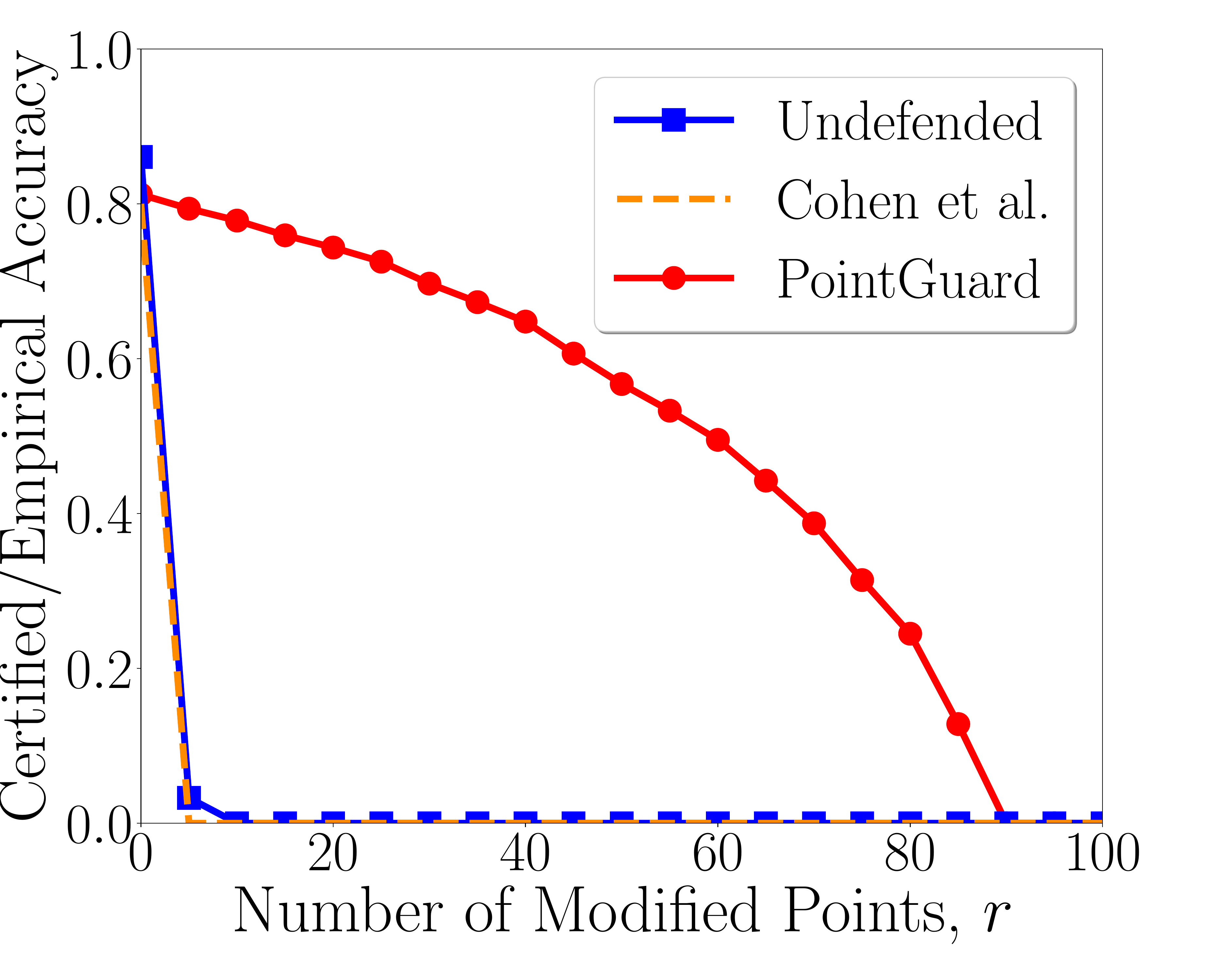}}
\subfloat[Point addition attacks.]{\includegraphics[width=0.24\textwidth]{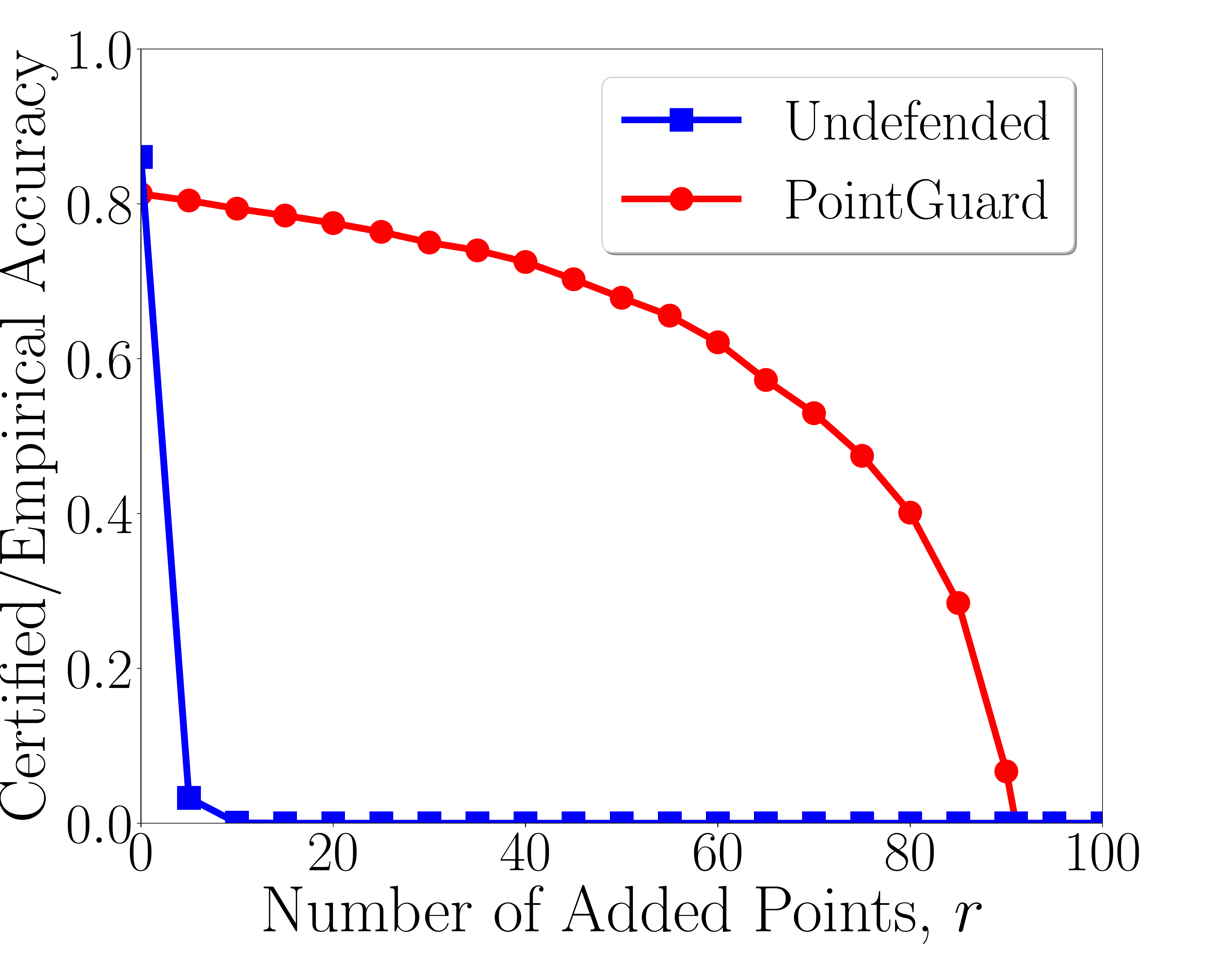}}
\subfloat[Point deletion attacks.]{\includegraphics[width=0.24\textwidth]{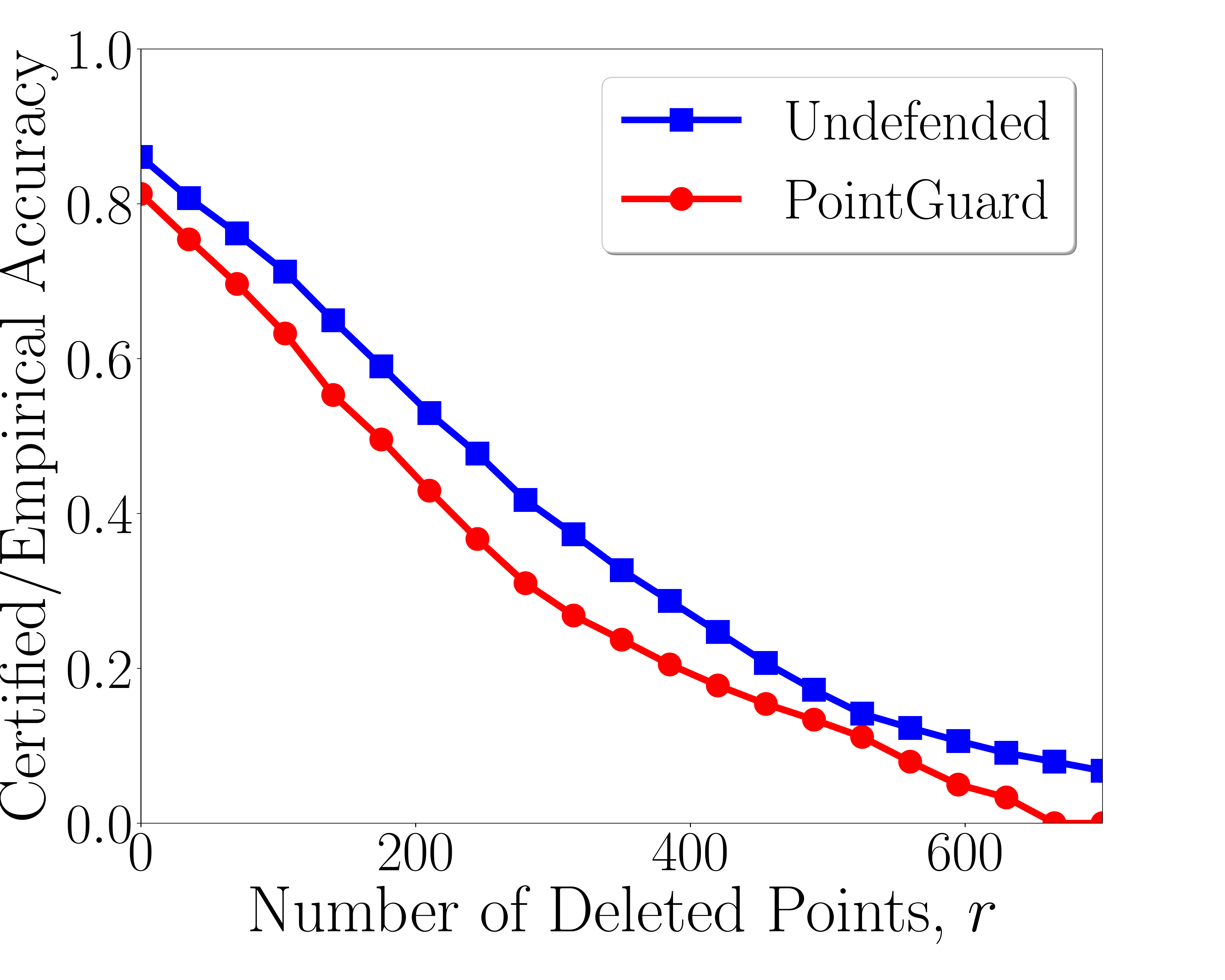}}

\caption{Comparing different methods under different attacks on ModelNet40. The results on ScanNet are in Supplemental Material.}
\label{fig:1}
\vspace{-2mm}
\end{figure*}

\myparatight{Compared methods} We compare our PointGuard with the following methods: 

{\bf Undefended classifier.} We call the standard point cloud classifier \emph{undefended classifier}, e.g., the undefended classifiers are PointNet and DGCNN on the two datasets in our experiments, respectively.

{\bf Randomized smoothing (Cohen et al.)~\cite{cohen2019certified}.} Randomized smoothing adds isotropic Gaussian noise with mean 0 and standard deviation $\sigma$ to an image before using a classifier to classify it. Randomized smoothing provably predicts the same label for the image when the $\ell_2$-norm of the adversarial perturbation added to the image is less than a threshold, which is called \emph{certified radius}. We can generalize randomized smoothing to certify robustness against point modification attacks. In particular, we can add Gaussian noise to each dimension of each point in a point cloud before using a point cloud classifier to classify it, and randomized smoothing provably predicts the same label for the point cloud when the $\ell_2$-norm of the adversarial perturbation is less than the certified radius. Note that our certified perturbation size is the number points that are perturbed by an attacker. Therefore, we transform the certified radius to certified perturbation size as follows. Suppose the points in the point clouds lie in the space $\Theta$, and the $\ell_2$-norm distance between two arbitrary points in the space $\Theta$ is no larger than  $\lambda$, i.e., $\max_{\theta_1, \theta_2}\lnorm{\theta_1 -\theta_2}_2 \leq \lambda$. For instance, $\lambda = 2\sqrt{3}$ for ModelNet40 and $\lambda=\sqrt{15}$ for ScanNet.  Then, we can employ the relationship between $\ell_0$-norm and $\ell_2$-norm to derive the certified perturbation size based on the certified radius returned by randomized smoothing. Specifically, given $\lambda$ and the $\ell_2$-norm certified radius $\delta$ for a point cloud, the certified perturbation size can be computed as $\lfloor \frac{\delta^2}{\lambda^2} \rfloor$. 

\begin{figure*}[!t]
 \vspace{-2mm}
\centering
\subfloat[Training with vs w/o subsampling]{\includegraphics[width=0.24\textwidth]{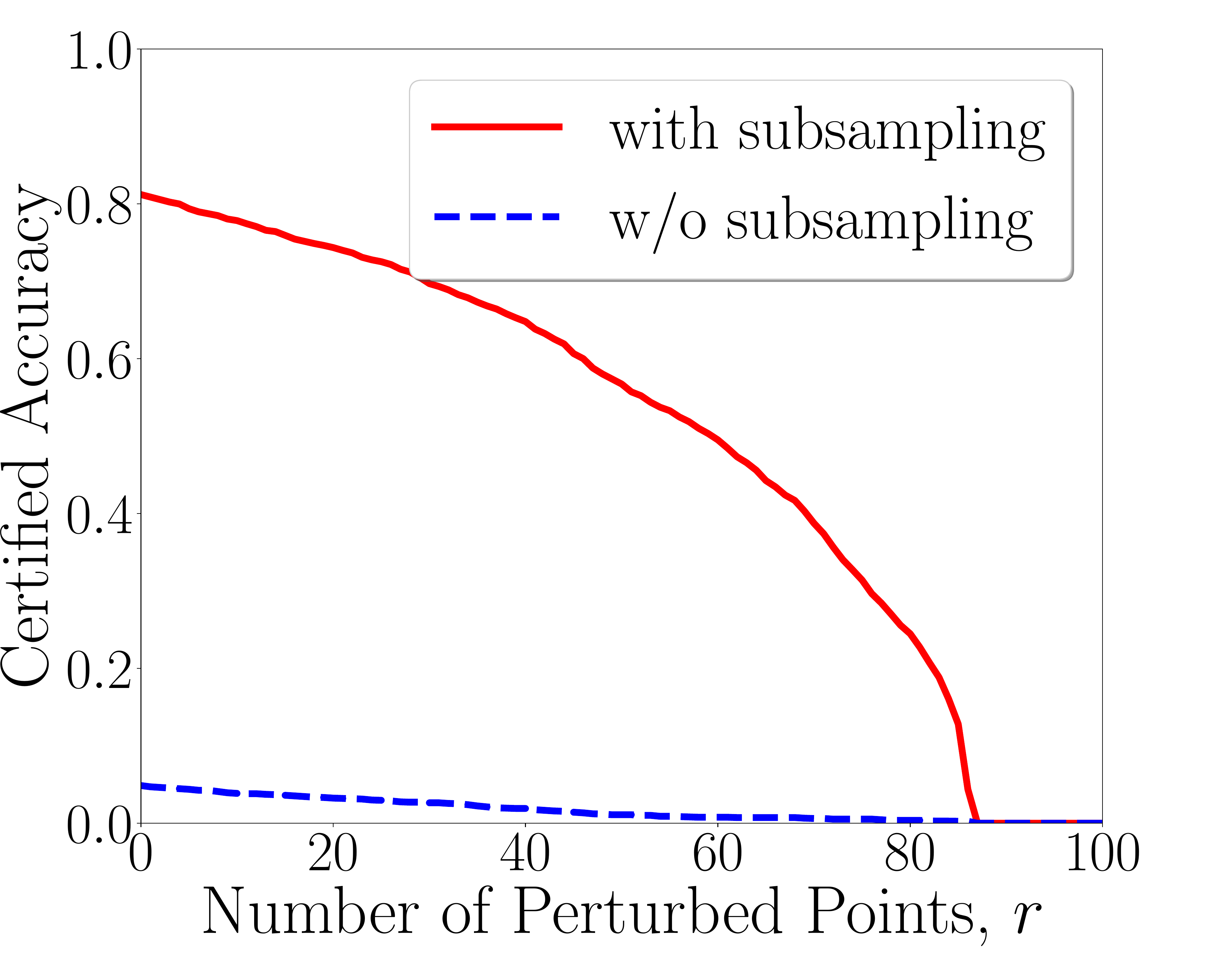}\label{impact_of_subsampling}}
\subfloat[Impact of $k$]{\includegraphics[width=0.24\textwidth]{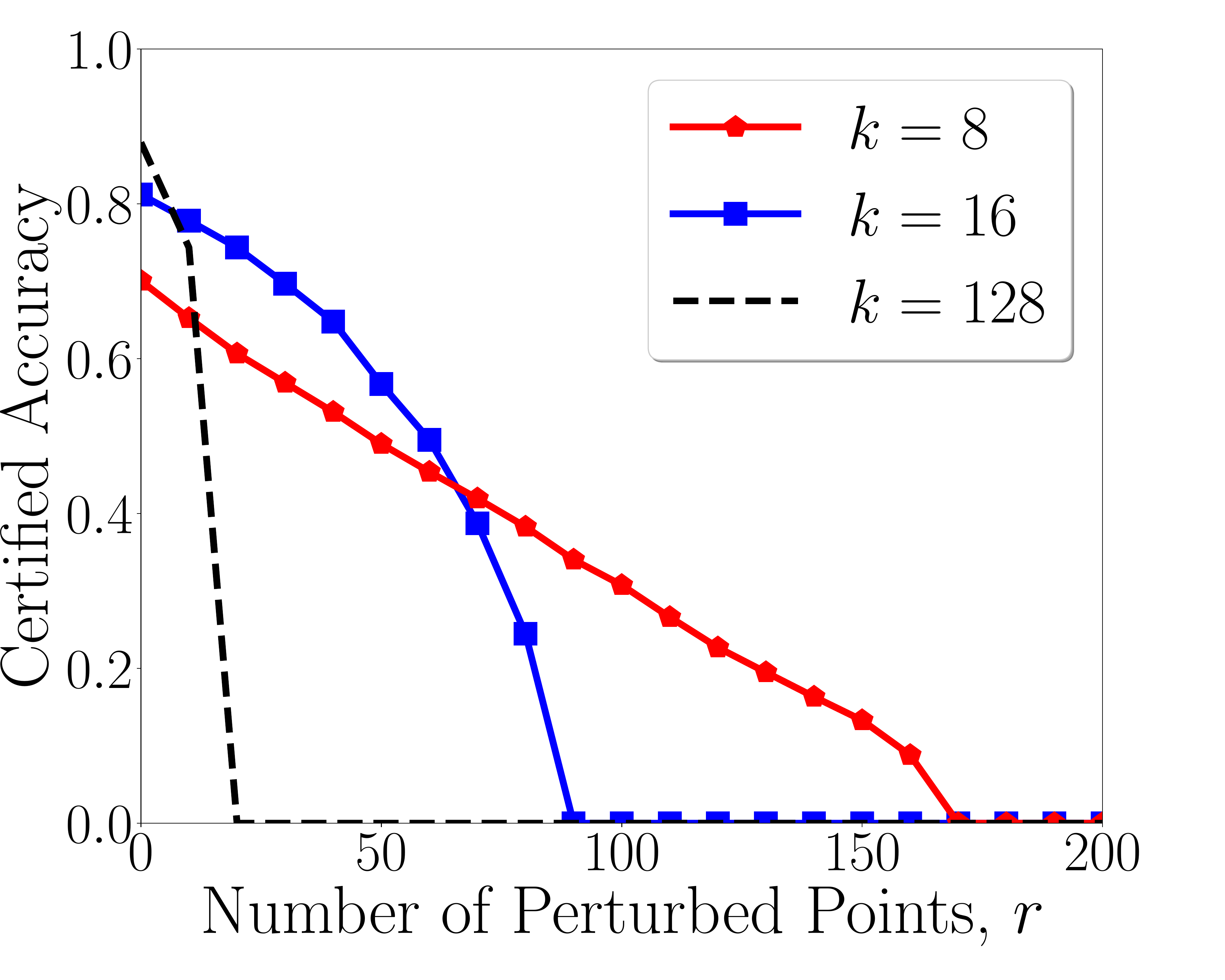}\label{impact_of_k}}
\subfloat[Impact of $\alpha$]{\includegraphics[width=0.24\textwidth]{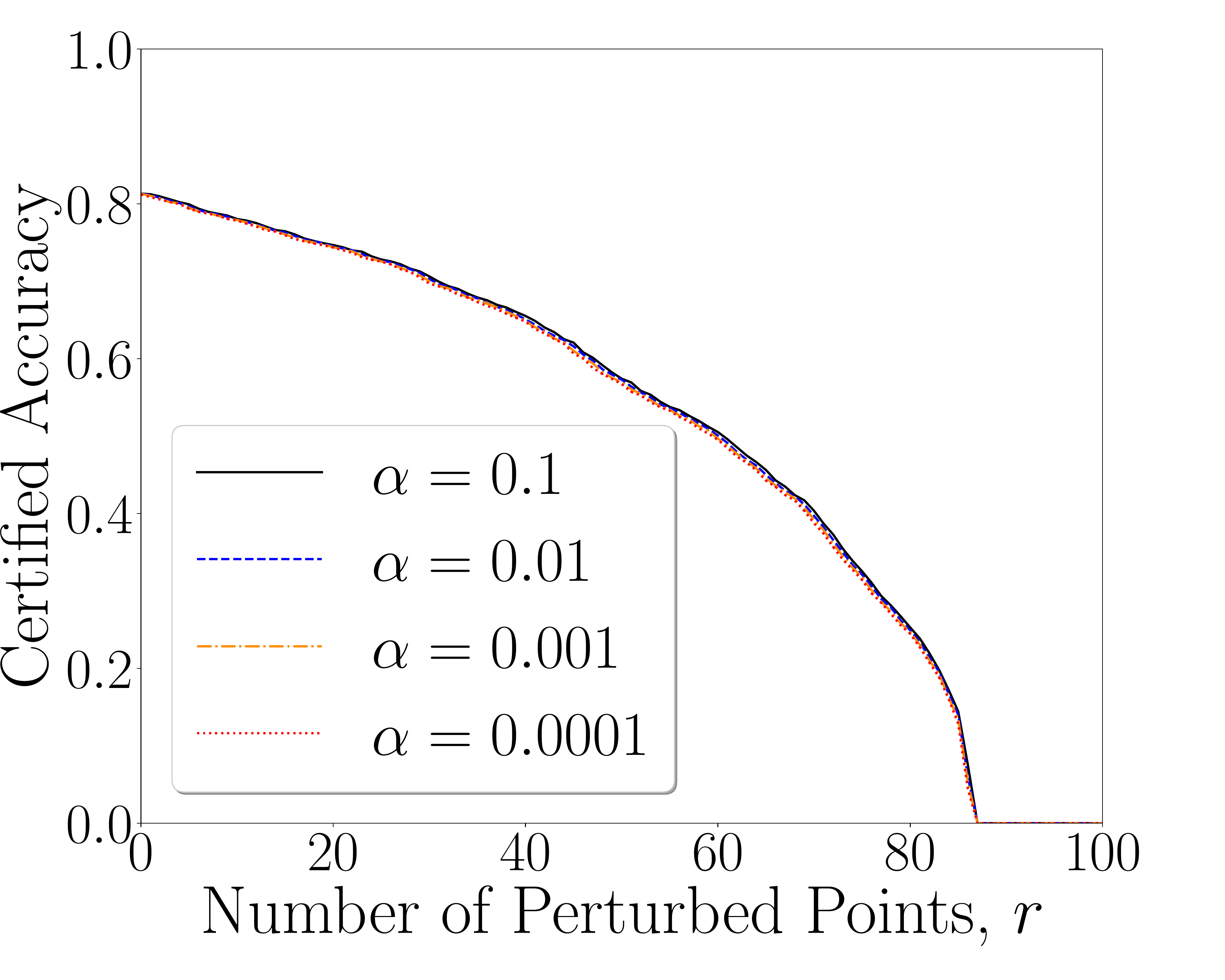}\label{impact_of_alpha}}
\subfloat[Impact of $N$]{\includegraphics[width=0.24\textwidth]{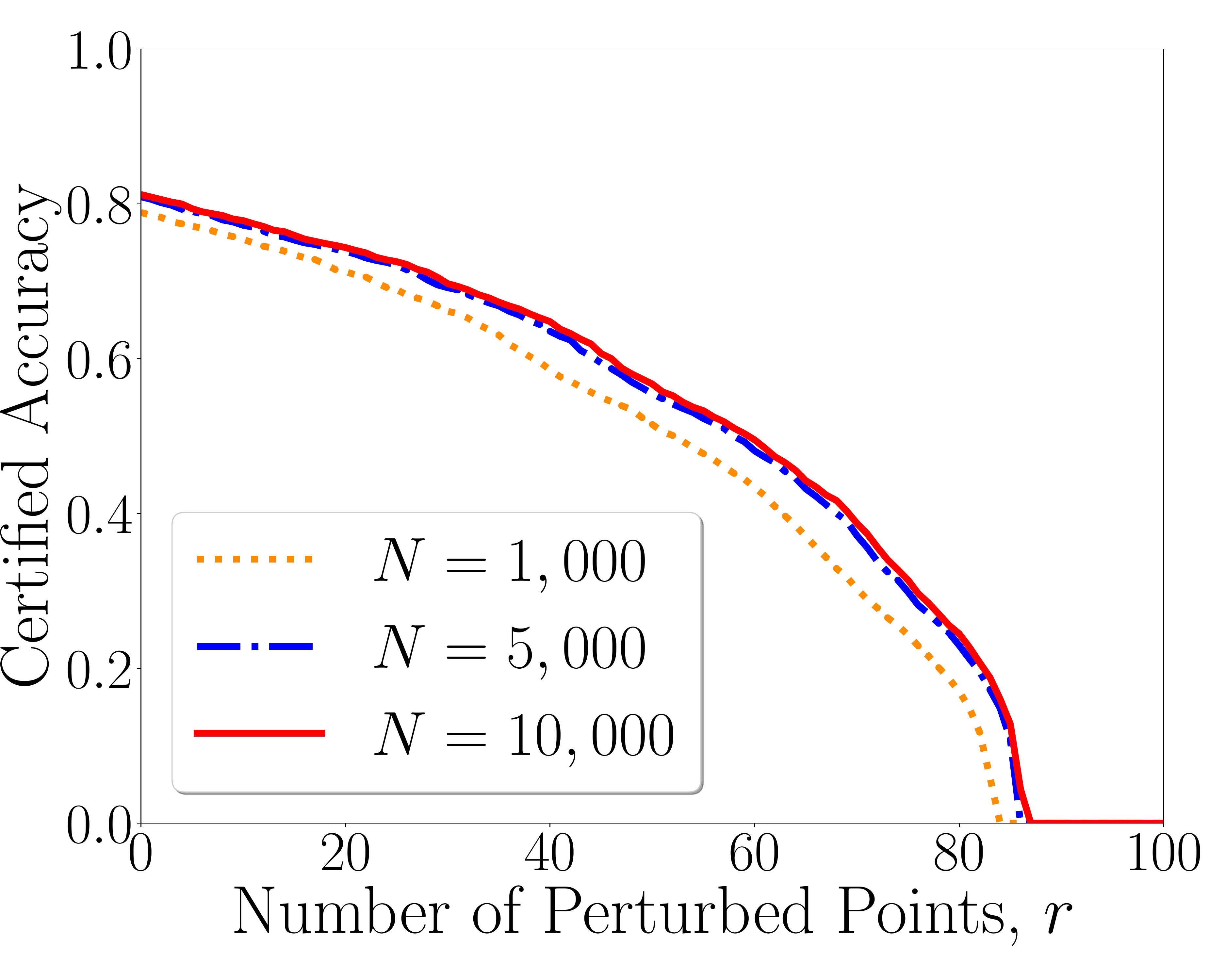}\label{impact_of_N}}

\caption{(a) Training the point cloud classifier with vs. without subsampling. (b), (c), and (d) show the impact of $k$, $\alpha$, and $N$, respectively. The dataset is ModelNet40. The results on ScanNet are in Supplemental Material.}
\label{fig:2}
\vspace{-2mm}
\end{figure*}

\myparatight{Evaluation metric}
PointGuard and randomized smoothing provide certified robustness guarantees, while the undefended classifiers provide empirical robustness. Therefore, we use \emph{certified accuracy}, which has been widely used to measure the certified robustness of a machine learning classifier against adversarial perturbations, to evaluate PointGuard and randomized smoothing; and we use \emph{empirical accuracy} under an emprical attack to evaluate the undefended classifiers. 
In  particular,  the  certified  accuracy  at $r$ perturbed points is the fraction of the testing point clouds whose  labels  are  correctly  predicted  and  whose  certified perturbation sizes are no smaller than $r$. Formally, given a testing set of point clouds $\mathcal{T}=\{(T_o,l_o)\}_{o=1}^{m}$, the certified accuracy at $r$ perturbed points is defined as follows:
\begin{align}
    CA_r = \frac{\sum_{o=1}^m \mathbb{I}\left(l_o=y_{o}\right) \cdot \mathbb{I}\left(r^{*}_o \geq r\right)}{m} ,
\end{align}
where $\mathbb{I}$ is an indicator function, $l_o$ is the true label of the testing point cloud $T_o$, and $y_o$ and $r^*_o$ respectively are the predicted label and the certified perturbation size returned by PointGuard or randomized smoothing for $T_o$. 

We use an empirical attack to calculate the empirical accuracy of an undefended classifier. Specifically, we first use the empirical attack to perturb each point cloud in a testing set, and then we use the undefended classifier to classify the perturbed point clouds and compute the accuracy (called empirical accuracy). 
    We adopt the empirical attacks proposed by Xiang et al.~\cite{xiang2019generating} for the point perturbation, modification, and addition attacks.  We use the empirical attack proposed by Wicker et al.~\cite{wicker2019robustness} for the point deletion attacks. 

We note that the certified accuracy at $r$ is a \emph{lower bound} of the accuracy that  PointGuard or randomized smoothing can achieve no matter how an attacker modifies, adds, and/or deletes at most $r$ points for each point cloud in the testing set, while the empirical accuracy under an empirical attack is an \emph{upper bound} of the accuracy that an undefended classifier can achieve under attacks.

\myparatight{Parameter setting}
Our PointGuard has three parameters: $k$, $1 - \alpha$, and $N$. 
Unless otherwise mentioned, we adopt the following default parameters: $\alpha = 0.0001$, $N = 10,000$, and $k = 16$ for both ModelNet40 and ScanNet. By default, we train the point cloud classifiers with subsampling. We adopt $\sigma = 0.5$ for randomized smoothing so its certified accuracy without attack is similar to that of PointGuard. By default, we consider the point perturbation attacks, as they are the strongest among the four types of attacks. 

\subsection{Experimental Results}
\vspace{-2mm}
\myparatight{Comparing PointGuard with other methods under different attacks}
Figure~\ref{fig:1} (or Figure~\ref{fig:scannet_1} in Supplemental Material) compares PointGuard with other methods under the four types of attacks  on ModelNet40 (or ScanNet), where an undefended classifier is measured by its empirical accuracy, while PointGuard and randomized smoothing are measured by certified accuracy.  We observe that the certified accuracy of PointGuard is slightly lower than the empirical accuracy of an undefended classifier when there are no attacks, i.e., $r=0$. However, for point perturbation, modification, and addition attacks, the empirical accuracy of an undefended classifier quickly drops to 0 while the certified accuracy of PointGuard is still high as $r$ increases. For point deletion attacks, the empirical accuracy of an undefended classifier may be higher than the certified accuracy of PointGuard. This indicates that the existing empirical point deletion attacks are not strong enough. 

Our PointGuard substantially outperforms randomized smoothing for point modification attacks in terms of certified accuracy. Randomized smoothing adds additive noise to a point cloud, while our PointGuard subsamples a point cloud. Our experimental results show that subsampling outperforms additive noise to build provably robust point cloud classification systems. We also observe that the empirical accuracy of the undefended classifier is close to   the certified accuracy of randomized smoothing, indicating that the empirical point modification attacks are strong. 

We also compare  PointGuard with an empirical defense (i.e., DUP-Net~\cite{zhou2019dup}) to measure the gaps between  certified accuracy and empirical accuracy. According to~\cite{zhou2019dup}, on ModelNet40, the empirical accuracy of DUP-Net under a point deletion attack~\cite{zhao2020isometry} is 76.1\%, 67.7\%, and 57.7\% when the attacker deletes 50, 100, and 150
points, respectively. Under the same setting,  PointGuard achieves certified accuracy of 73.4\%, 64.3\%, and 53.5\%, respectively. We observe the gaps between the empirical accuracy (an upper bound of accuracy) of DUP-Net and certified accuracy (a lower bound of accuracy) of PointGuard are small.

\myparatight{Training with vs. without subsampling}
Figure~\ref{impact_of_subsampling} (or Figure~\ref{impact_of_subsampling_1} in Supplemental Material) shows the comparison of the certified accuracy of PointGuard when the point cloud classifier is trained with or without subsampling  on ModelNet40 (or ScanNet). Our experimental results demonstrate that training with subsampling can substantially improve the certified accuracy of PointGuard. The reason is that the point cloud classifier trained with subsampling can more accurately classify the subsampled point clouds. 

\myparatight{Impact of $k$, $\alpha$, and $N$}
Figure~\ref{impact_of_k},~\ref{impact_of_alpha}, and~\ref{impact_of_N} (or Figure~\ref{impact_of_k_1},~\ref{impact_of_alpha_1}, and~\ref{impact_of_N_1} in Supplemental Material) show the impact of $k$, $\alpha$, and $N$ on the certified accuracy of our PointGuard  on ModelNet40 (or ScanNet), respectively. Based on the experimental results, we make the following observations. First, $k$ measures a tradeoff between accuracy without attacks (i.e., $r=0$) and robustness. In particular, a smaller $k$ gives a smaller certified accuracy without attacks, but the certified accuracy drops to 0 more slowly as the number of perturbed points $r$ increases. The reason is that the perturbed points are less likely to be subsampled when $k$ is smaller.  Second, the certified accuracy increases as $\alpha$ or $N$ increases. The reason is that a larger $\alpha$ or $N$ leads to tighter  lower and upper label probability bounds, which in turn lead to larger certified perturbation sizes. We also note that the certified accuracy is insensitive to $\alpha$ and $N$ once they are large enough.  

\section{Conclusion}
In this work, we propose PointGuard, the first provably robust 3D point cloud classification system against various adversarial attacks. We show that PointGuard provably predicts the same label for a testing 3D point cloud when the number of adversarially perturbed points is bounded. Moreover, we prove our bound is tight.  We empirically demonstrate the effectiveness of PointGuard on ModelNet40 and ScanNet benchmark datasets. An interesting future work is to further improve  PointGuard by leveraging the knowledge of the point cloud classifier.

{\small
\bibliographystyle{ieee_fullname}
\bibliography{egbib}
}

\appendices
\clearpage

\section{Proof of Theorem~\ref{certified_radius_pointcloud}}
\label{proof_of_certified_radius_bagging}

Given a point cloud $T$ and its perturbed version $T^*$, we define the following two random variables:
\begin{align}
   & \bm W = S_k(T),~\bm Z = S_k(T^*), 
\end{align}
where $\bm W$ and $\bm Z$ represent the random 3D point clouds with $k$ points subsampled from $T$ and $T^*$ uniformly at random without replacement, respectively. We use $\Phi$ to denote the joint space of $\mathbf{W}$ and $\mathbf{Z}$, where each element is a 3D point cloud with $k$ points subsampled from $T$ or $T^*$. We denote by $E$ the set of intersection points between $T$ and $T^*$, i.e., $E = T \cap T^*$. 

Before proving our theorem, we first describe a variant of the Neyman-Pearson Lemma~\cite{neyman1933ix} that will be used in our proof. The variant is from~\cite{jia2020intrinsic}.

\begin{lemma}[Neyman-Pearson Lemma]
\label{lemma_np}
Suppose $\bm W$ and $\bm Z$ are two random variables in the space $\Phi$ with probability distributions $\sigma_{w}$ and $\sigma_{z}$, respectively. Let $H:\Phi\xrightarrow{} \{0,1\}$ be a random or deterministic function. Then, we have the following:  
\begin{itemize}
\item If $Q_1=\{\varphi\in \Phi:\sigma_{w}(\varphi) > \zeta \cdot \sigma_{z}(\varphi)  \}$ and $Q_2=\{\varphi\in \Phi:\sigma_{w}(\varphi) = \zeta \cdot \sigma_{z}(\varphi)  \}$ for some $\zeta > 0$. Let $Q=Q_1\cup Q_3$, where $Q_3 \subseteq Q_2$. If we have $\text{Pr}(H(\bm W)=1)\geq \text{Pr}(\bm W\in Q)$, then $\text{Pr}(H(\bm Z)=1)\geq \text{Pr}(\bm Z\in Q)$.

\item If $Q_1=\{\varphi\in \Phi:\sigma_{w}(\varphi) < \zeta\cdot \sigma_{z}(\varphi)  \}$ and $Q_2=\{\varphi\in \Phi:\sigma_{w}(\varphi) = \zeta\cdot \sigma_{z}(\varphi)  \}$ for some $\zeta > 0$. Let $Q=Q_1\cup Q_3$, where $Q_3 \subseteq Q_2$. If we have $\text{Pr}(H(\bm W)=1)\leq \text{Pr}(\bm W\in Q)$, then $\text{Pr}(H(\bm Z)=1)\leq \text{Pr}(\bm Z\in Q)$.
\end{itemize}
\end{lemma}

\begin{proof}
Please refer to~\cite{jia2020intrinsic}.
\end{proof}

Next, we formally prove our Theorem~\ref{certified_radius_pointcloud}. Our proof is inspired by previous work~\cite{jia2019certified,jia2020intrinsic}. Roughly speaking, the idea is to derive the label probability lower and upper bounds via computing the probability of random variables in certain regions crafted by the variant of Neyman-Pearson Lemma. However, due to the difference in sampling methods, our space divisions are significantly different from previous work~\cite{jia2020intrinsic}. Recall that we denote $p_i =  \text{Pr}(f(\mathbf{W})=i)$ and $p_i^{*} = \text{Pr}(f(\mathbf{Z})=i)$, where $i \in \{1, 2, \cdots, c\}$. We denote $y = \argmax_{i=\{1,2,\dots,c\}} p_i$. Our goal is to find the maximum $r^*$ such that $y =  \argmax_{i=\{1, 2, \cdots, c\}} p_i^{*}$, i.e., $p_{y}^{*} > p_e^* = max_{i \neq y} p_{i}^{*}$, for $\forall T^{*} \in \Gamma(T, r^*)$.  Our key step is to derive a lower bound of $p_{y}^{*} $ and an upper bound of $p_e^* = max_{i \neq y} p_{i}^{*}$ via Lemma~\ref{lemma_np}. Given these probability bounds, we can find the maximum $r^*$ such that the lower bound of $p_{y}^{*} $ is larger than the upper bound of $p_e^*$. 

\myparatight{Dividing the space $\Phi$} We first divide the space $\Phi$  into three regions which are as follows: 
\begin{align}
\Delta_T~=\{\varphi\in\Phi & |\varphi \subseteq T, \varphi\not\subseteq E  \}, \\
 \Delta_{T^*}=\{\varphi\in\Phi& |\varphi \subseteq T^*, \varphi\not\subseteq E  \}, \\
\Delta_E~=\{\varphi\in\Phi & | \varphi \subseteq E  \}, 
\end{align}
where $\Delta_E$ consists of the subsampled point clouds that can be obtained by subsampling $k$ points from $E$; and $\Delta_T$ (or $\Delta_{T^{*}}$) consists of the subsampled point clouds that are subsampled from $T$ (or $T^*$) but do not belong to $\Delta_E$. 
Since $\bm W$ and $\bm Z$, respectively, represent the random 3D point clouds with $k$ points subsampled from $T$  and $T^*$ uniformly at random without replacement, we have the following probability mass functions: 
\begin{align}
&\text{Pr}(\bm W = \varphi) = 
\begin{cases}
 \frac{1}{{n \choose k}}, &\text{ if } \varphi\in \Delta_T  \cup \Delta_E, \\
 0, &\text{ otherwise},
\end{cases} \\
&\text{Pr}(\bm Z = \varphi) = 
\begin{cases}
 \frac{1}{{t \choose k}}, &\text{ if } \varphi\in \Delta_{T^*} \cup \Delta_E, \\
 0, &\text{ otherwise},
\end{cases}
\end{align}
where $t$ is the number of points in $T^{*}$ (i.e., $t = |T^*|$). We use $s$ to denote the number of intersection points between $T$ and $T^*$, i.e., $s = |E|=|T \cap T^*|$. Then, the size of $\Delta_E$ is ${s \choose k}$, i.e., $|\Delta_E|={s \choose k}$. Given the size of $\Delta_E$, we have the following probabilities: 
{\small 
\begin{align}
\label{probability_equation_1_bagging}
&\text{Pr}(\bm W \in \Delta_E)= \frac{{s \choose k}}{{n \choose k}}, \\ &\text{Pr}(\bm W \in \Delta_T)=1-  \frac{{s \choose k}}{{n \choose k}}, \\ &\text{Pr}(\bm W \in \Delta_{T^*})=0. \\
&\text{Pr}(\bm Z \in \Delta_E)=  \frac{{s \choose k}}{{t \choose k}} ,\\ &\text{Pr}(\bm Z \in \Delta_{T^*})=1-  \frac{{s \choose k}}{{t \choose k}},\\
\label{probability_equation_2_bagging}
&\text{Pr}(\bm Z \in \Delta_T)=0.
\end{align}
}
We have $\text{Pr}(\bm W \in \Delta_E)= \frac{{s \choose k}}{{n \choose k}}$  because $\text{Pr}(\bm W \in \Delta_E)= \frac{|\Delta_E|}{|\Delta_T \cup \Delta_E|} = \frac{{s \choose k}}{ {n \choose k} }$. Since $\text{Pr}(\bm W \in \Delta_T)+ \text{Pr}(\bm W \in \Delta_E) = 1$, we have $\text{Pr}(\bm W \in \Delta_T)=1-  \frac{{s \choose k}}{{n \choose k}}$. We have $\text{Pr}(\bm W \in \Delta_{T^*})=0$ because $\bm W $ is subsampled from $T$, which does not contain any points from $ T^* \setminus E$. Similarly, we can compute the probabilities of random variable $\bm Z$ in those regions. 

Based on the fact that $p_y$ and $p_i (i \neq y)$ should be integer multiples of $1/{n \choose k}$, we derive the following bounds:
\begin{align}
\label{label_probability_bound_app_1}
& \underline{p^{\prime}_y} \triangleq \frac{\lceil \underline{p_y}\cdot{n \choose k} \rceil}{{n \choose k}} \leq  \text{Pr}(f(\mathbf{W})= y), \\
\label{label_probability_bound_app_2}
& \overline{p}^{\prime}_i \triangleq \frac{\lfloor \overline{p}_i \cdot{n \choose k} \rfloor}{{n \choose k}}  \geq \text{Pr}(f(\mathbf{W})=i), \forall i \neq y.
\end{align}

{\bf Deriving a lower bound of $p^*_y$:} We define a binary function $H_y(\varphi)= \mathbb{I}(f(\varphi)=y)$, where $\varphi\in \Phi$ and $\mathbb{I}$ is an indicator function. Then, we have the following  based on the definitions of the random variable $\bm Z$ and the function $H_y$: 
\begin{align}
p^*_y=\text{Pr}(f(\bm Z)=y)=\text{Pr}(H_y(\bm Z)=1).
\end{align}

Our idea is to find a region such that we can apply Lemma~\ref{lemma_np} to derive a lower bound of $\text{Pr}(H_y(\bm Z)=1)$. 
We assume $\underline{p^{\prime}_y} - \left(1- \frac{{s \choose k}}{{n \choose k}}\right)\geq 0$.  We can make this assumption because we only need to find a sufficient condition. Then, we can find a region $\Delta_y \subseteq \Delta_E$ satisfying the following: 
\begin{align}
\label{definition_of_aprime_pointcloud_1}
   & \text{Pr}(\bm W\in\Delta_y) \\
    =& \underline{p^{\prime}_y} - \text{Pr}(\bm W\in\Delta_T) \\ 
    \label{definition_of_aprime_pointcloud_2}
    =& \underline{p^{\prime}_y}- \left(1- \frac{{s \choose k}}{{n \choose k}}\right). 
\end{align}
We can find the region $\Delta_y$ because $\underline{p^{\prime}_y}$ is an integer multiple of $\frac{1}{{n \choose k}}$. Given the region $\Delta_y$, we define the following region: 
\begin{align}
\label{definition_of_r_bagging}
    \mathcal{A} = \Delta_T\cup \Delta_y.
\end{align}
Then, based on Equation~(\ref{label_probability_bound_app_1}), we have:
\begin{align}
  \text{Pr}(f(\bm W) = y )   \geq \underline{p^{\prime}_y}= \text{Pr}(\bm W \in \mathcal{A}).
\end{align}
We can establish the following based on the definition of $\bm W$: 
\begin{align}
\label{lemma_np_condition_ab_bagging}
  \text{Pr}(H_y(\bm W)=1)=   \text{Pr}(f(\bm W)=y) \geq \text{Pr}(\bm W \in \mathcal{A}). 
\end{align}
Furthermore, we have $\text{Pr}(\bm W = \varphi) > \epsilon  \cdot \text{Pr}(\bm Z = \varphi)$ if and only if $\varphi\in \Delta_T$ and $\text{Pr}(\bm W = \varphi) = \epsilon  \cdot \text{Pr}(\bm Z = \varphi)$ if $\varphi\in \Delta_y$, where $\epsilon = \frac{{t \choose k}}{{n \choose k}}$. Therefore, based on the definition of $\mathcal{A}$ in Equation~(\ref{definition_of_r_bagging}) and the condition in Equation~(\ref{lemma_np_condition_ab_bagging}), we obtain the following by applying Lemma~\ref{lemma_np}: 
\begin{align}
\text{Pr}(H_y(\bm Z)=1) \geq \text{Pr}(\bm Z \in \mathcal{A}).
\end{align}
Since we have $p^*_y = \text{Pr}(H_y(\bm Z)=1)$, $\text{Pr}(\bm Z \in \mathcal{A})$ is a lower bound of $p^*_y$ and can be computed as follows: 
\begin{align}
 &\text{Pr}(\bm Z \in \mathcal{A}) \\
 \label{probability_vk_in_r_1_bagging}
 =& \text{Pr}(\bm Z \in \Delta_T) + \text{Pr}(\bm Z \in \Delta_y) \\
  \label{probability_vk_in_r_2_bagging}
 =& \text{Pr}(\bm Z \in \Delta_y) \\
  \label{probability_vk_in_r_3_bagging}
 =& \text{Pr}(\bm W \in \Delta_y)/\epsilon  \\
  \label{probability_vk_in_r_4_bagging}
 =& \frac{1}{\epsilon }\cdot \left(\underline{p^{\prime}_y} -  \left(1- \frac{{s \choose k}}{{n \choose k}}\right)\right).
\end{align}
We have Equation~(\ref{probability_vk_in_r_2_bagging}) from~(\ref{probability_vk_in_r_1_bagging}) because $\text{Pr}(\bm Z \in \Delta_T)=0$, Equation~(\ref{probability_vk_in_r_3_bagging}) from~(\ref{probability_vk_in_r_2_bagging}) as $\text{Pr}(\bm W = \varphi) = \epsilon  \cdot \text{Pr}(\bm Z = \varphi)$ for $\varphi\in \Delta_y$, and the last Equation from Equation~(\ref{definition_of_aprime_pointcloud_1})~-~(\ref{definition_of_aprime_pointcloud_2}).

{\bf Deriving an upper bound of $\max_{i\neq y} p^*_i$:} 
We leverage the second part of Lemma~\ref{lemma_np} to derive an upper bound of $\max_{i\neq y} p^*_i$. We assume $\text{Pr}(\bm W \in\Delta_E) > \overline{p}^{\prime}_i $, $\forall i \in \{1,2,\cdots,c\}\setminus \{y\}$. We can make the assumption because we aim to  derive a sufficient condition.
For $\forall i \in \{1,2,\cdots,c\}\setminus \{y\}$, we can find a region $\Delta_i\subseteq \Delta_E$ such that we have the following: 
\begin{align}
\label{c_j_condition_in_certify}
    \text{Pr}(\bm W \in\Delta_i) = \overline{p}^{\prime}_i. 
\end{align}
We can find the region because $\overline{p}^{\prime}_i $ is an integer multiple of $\frac{1}{{n \choose k}}$. Given region $\Delta_i$, we define the following region: 
\begin{align}
\label{definition_of_q_bagging}
    \mathcal{B}_i = \Delta_i \cup \Delta_{T^*} .
\end{align}
For $\forall i \in \{1, 2, \cdots, c\}\setminus \{y\}$, we define a function $H_i(\varphi)= \mathbb{I}(f(\varphi)=i)$, where $\varphi\in \Phi$. Then, based on Equation~(\ref{label_probability_bound_app_2}) and the definition of random variable $\bm W$, we have:
\begin{align}
\label{region_delta_i_condition_app}
\text{Pr}(H_i(\bm W)=1) =  \text{Pr}(f(\bm W) = i)   \leq \overline{p}^{\prime}_i = \text{Pr}(\bm W \in \mathcal{B}_i). 
\end{align}
We note that $\text{Pr}(\bm W = \varphi) < \epsilon  \cdot \text{Pr}(\bm Z = \varphi)$ if and only if $\varphi\in \Delta_{T^*}$ and $\text{Pr}(\bm W = \varphi) = \epsilon  \cdot \text{Pr}(\bm Z = \varphi)$ if $\varphi\in \Delta_i$, where $\epsilon = \frac{{t \choose k}}{{n \choose k}}$.
Based on  the definition of random variable $\bm Z$, Equation~(\ref{region_delta_i_condition_app}), and Lemma~\ref{lemma_np}, we have the following: 
\begin{align}
    \text{Pr}(H_i (\bm Z)=1) \leq \text{Pr}(\bm Z \in \mathcal{B}_i).
\end{align}
Since we have $p^*_i =\text{Pr}(f(\bm Z)=i)= \text{Pr}(H_i (\bm Z)=1)$, $\text{Pr}(\bm Z \in \mathcal{B}_i)$ is an upper bound of $p^*_i$ and  can be computed as follows: 
\begin{align}
 &\text{Pr}(\bm Z \in \mathcal{B}_i) \\
 =& \text{Pr}(\bm Z \in \Delta_i) + \text{Pr}(\bm Z \in \Delta_{T^*})  \\
 =& \text{Pr}(\bm Z \in \Delta_i)+ 1- \frac{{s \choose k}}{{t\choose k}}   \\
 =& \text{Pr}(\bm W \in \Delta_i)/\epsilon + 1- \frac{{s \choose k}}{{t\choose k}}   \\
 =& \frac{1}{\epsilon }\cdot \overline{p}^{\prime}_i + 1- \frac{{s \choose k}}{{t\choose k}} .
\end{align}
By considering all possible $i$ in the set $\{1,2,\cdots,c\}\setminus \{y\}$, we have: 
\begin{align}
 &\max_{i\neq y} p^*_i \\
 \leq & \max_{i\neq y} \text{Pr}(\bm Z \in \mathcal{B}_i) \\
 =&  \frac{1}{\epsilon }\cdot \max_{i\neq y} \overline{p}^{\prime}_i + 1- \frac{{s \choose k}}{{t\choose k}}  \\
 \leq &  \frac{1}{\epsilon }\cdot \overline{p}^{\prime}_e + 1- \frac{{s \choose k}}{{t\choose k}} ,
\end{align}
where $\overline{p}^{\prime}_e \geq \max\limits_{i \neq y} \overline{p}^{\prime}_i$. 

{\bf Deriving the certified perturbation size:} To reach our goal $\text{Pr}(f(\bm Z)=y) > \max\limits_{i\neq y}\text{Pr}(f(\bm Z)=i)$, it is sufficient to have the following:   
\begin{align}
  &  \frac{1}{\epsilon }\cdot \left(\underline{p^{\prime}_y} - \left(1- \frac{{s \choose k}}{{n \choose k}}\right)\right) 
> \frac{1}{\epsilon }\cdot \overline{p}^{\prime}_e + 1- \frac{{s \choose k}}{{t\choose k}} \\
\label{certify_poisoning_size_condition_proof_final_check_pre}
 \Longleftrightarrow &  \frac{{t \choose k}}{{n \choose k}} -2\cdot \frac{{s \choose k}}{{n \choose k}} + 1 - \underline{p^{\prime}_y} + \overline{p}^{\prime}_e < 0.
\end{align}
Since Equation~(\ref{certify_poisoning_size_condition_proof_final_check_pre}) should be satisfied for all possible perturbed point cloud $T^*$  (i.e., $n-r \leq t \leq  n+r$), we have the following sufficient condition: 
\begin{align}
\label{certify_poisoning_size_condition_proof_final_check}
  \max_{n-r \leq t \leq n+r} \frac{{t \choose k}}{{n \choose k}} -2\cdot \frac{{s \choose k}}{{n \choose k}} + 1 - \underline{p^{\prime}_y} + \overline{p}^{\prime}_e  < 0. 
\end{align}
When the above Equation~(\ref{certify_poisoning_size_condition_proof_final_check}) is satisfied, we have $\underline{p^{\prime}_y} - \left(1- \frac{{s \choose k}}{{n \choose k}}\right)\geq 0$ and $\text{Pr}(\bm W \in\Delta_E) =\frac{{s \choose k}}{{n \choose k}} \geq \overline{p}^{\prime}_i ,\forall i  \in \{1,2,\cdots,c\}\setminus \{y\}$, which are the conditions that we rely on to construct the region $\Delta_y$ and $\Delta_i (i \neq y)$. The certified perturbation size $r^{*}$ is the maximum $r$ that satisfies the above sufficient condition. 
Note that $s = \max(n,t) - r$. Then, our certified perturbation size $r^{*}$ can be derived by solving the following optimization problem: 
\begin{align}
r^{*} & =  \argmax_{r} r \nonumber \\
\label{certified_condition_bagging_proof}
& \text{s.t.} \max_{n-r \leq t \leq n+r} \frac{{t \choose k}}{{n \choose k}} -2\cdot \frac{{\max(n,t) - r \choose k}}{{n \choose k}} + 1 - \underline{p^{\prime}_y} + \overline{p}^{\prime}_e < 0.  
\end{align}

\section{Proof of Theorem~\ref{tightness}}
\label{proof_of_tightness}
Similar to previous work~\cite{cohen2019certified,jia2019certified,jia2020intrinsic}, we show the tightness of our bounds via constructing a counterexample.
In particular, when $r > r^*$, we will show that we can construct a point cloud $T^{*}$ and a point cloud classifier $f^*$  which satisfies the Equation~(\ref{probability_bound_main}) such that the label $y$ is not predicted by our PointGuard or there exist ties. Since $r^*$ is the maximum value that satisfies the Equation~(\ref{certified_condition_bagging_proof}), there exists a point cloud $T^{*}$ satisfying the following when $r > r^*$:
{\small 
\begin{align}
&\frac{{t \choose k}}{{n \choose k}} -2\cdot \frac{{\max(n,t)-r \choose k}}{{n \choose k}} + 1 - \underline{p^{\prime}_y} + \overline{p}^{\prime}_e  \geq 0\\
\Longleftrightarrow & \frac{{t \choose k}}{{n\choose k }}  - \frac{{\max(n,t)-r \choose k}}{{n\choose k }} + \overline{p}^{\prime}_e \geq \underline{p^{\prime}_y} - \left(1-\frac{{\max(n,t)-r \choose k}}{{n\choose k }}\right)   \\
 \label{tight_main}
\Longleftrightarrow & \frac{\overline{p}^{\prime}_e}{\epsilon } + 1- \frac{{\max(n,t)-r \choose k}}{{t\choose k}}  \geq \frac{1}{\epsilon }\cdot \left(\underline{p^{\prime}_y} - \left(1- \frac{{\max(n,t)-r \choose k}}{{n \choose k}}\right)\right)
\end{align}
}
where $t$ is the number of points in $T^*$ and $\epsilon=\frac{{t \choose k}}{{n\choose k }}$. Let $\Delta_e \subseteq \Delta_E$ be the region that satisfies the following:
\begin{align}
\Delta_e \cap \Delta_y = \emptyset \text{ and } \Pr(\bm W \in \Delta_e) = \overline{p}^{\prime}_e.
\end{align}
Note that we can find the region $\Delta_e$ because we have $\underline{p^{\prime}_y}+\overline{p}^{\prime}_e \leq 1$. We let  $\mathcal{B}_e = \Delta_{T^*} \cup \Delta_e$.  Then, we can divide the the region $\Phi \setminus (\mathcal{A} \cap \mathcal{B}_e)$ into $c-2$ regions and we use $\mathcal{B}_i$ to denote each region, where $ i\in \{1,2,\cdots,c\}\setminus\{y,e\}$. In particular, each region $\mathcal{B}_i$ satisfies $\Pr(T \in \mathcal{B}_i) \leq \overline{p}^{\prime}_i$. We can find these regions since $\underline{p^{\prime}_y}+\sum_{i \neq y} \overline{p}^{\prime}_i \geq 1$. Then, we can construct the following point cloud classifier $f^*$:
\begin{align}
f^*(\varphi) = 
\begin{cases}
y, \text{ if } \varphi \in \mathcal{A}, \\
i, \text{ if } \varphi \in \mathcal{B}_i.
\end{cases}
\end{align}
Note that the point cloud classifier $f^*$ is well-defined in the space $\Phi$.
We have the following probabilities for our constructed point cloud classifier $f^*$:
\begin{align}
&\Pr(f^*(\bm W)=y) = \Pr(\bm W \in \mathcal{A}) = \underline{p^{\prime}_y},\\
&\Pr(f^*(\bm W)=e) = \Pr(\bm W \in \mathcal{B}_e) = \overline{p}^{\prime}_e,\\
&\Pr(f^*(\bm W)=i) = \Pr(\bm W \in \mathcal{B}_i) \leq \overline{p}^{\prime}_i,
\end{align}
where $i \in \{1,2,\dots,c\}\setminus\{y,e\}$. Note that the point cloud classifier $f^*$ is consistent with Equation~(\ref{probability_bound_main}). Moreover, we have the following:
\begin{align}
& \Pr(f^*(\bm Z)= e) \\
=& \Pr(\bm Z \in \mathcal{B}_e) \\
\label{tight_left}
= &  \frac{\overline{p}^{\prime}_e }{\epsilon } +1- \frac{{\max(n,t)-r \choose k}}{{t\choose k}} \\
\label{tight_right}
\geq & \frac{1}{\epsilon} \left( \underline{p^{\prime}_y} - \left(1-\frac{{\max(n,t)-r \choose k }}{{n \choose k }}\right) \right) \\
=& \Pr(\bm Z \in \mathcal{A}) \\ 
=&\Pr(f^*(\bm Z)=y),
\end{align}
where $\epsilon = \frac{{t \choose k}}{{n \choose k }}$. Note that we derive Equation (\ref{tight_right}) from (\ref{tight_left}) based on Equation~(\ref{tight_main}). Therefore, for $\forall r > r^*$, there exist a point cloud classifier $f^*$ which satisfies Equation~(\ref{probability_bound_main}) and a point cloud $T^{*}$ such that $g(T^{*}) \neq y$ or there exist ties.

\section{Acknowledgements}
We thank the anonymous reviewers for constructive reviews and comments. This work was partially supported by NSF grant No.1937786.

\begin{figure*}[!t]
 \vspace{-2mm}
\centering
\subfloat[Point perturbation attacks. ]{\includegraphics[width=0.24\textwidth]{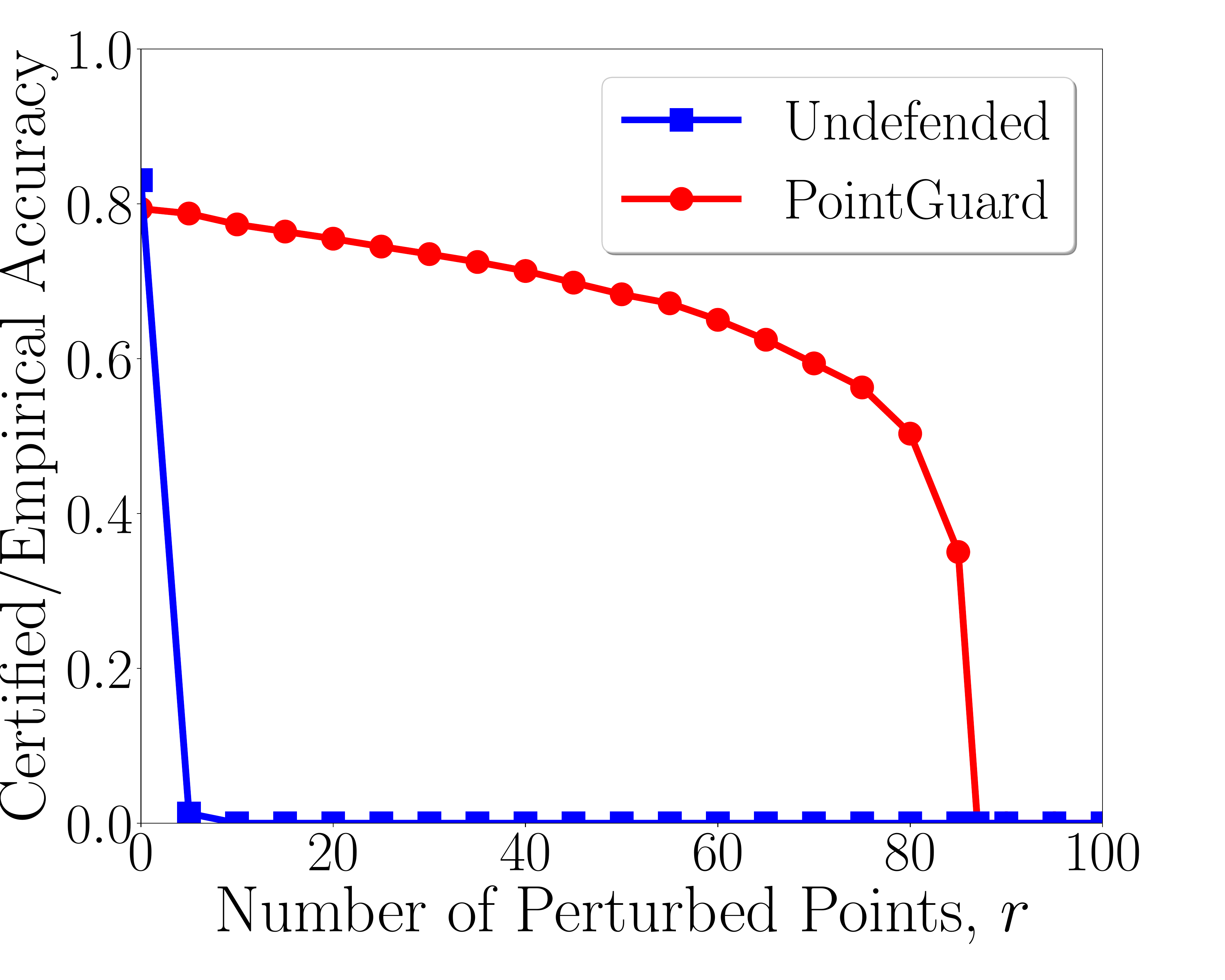}\label{compare_perturbation_attack}}
\subfloat[Point modification attacks.]{\includegraphics[width=0.24\textwidth]{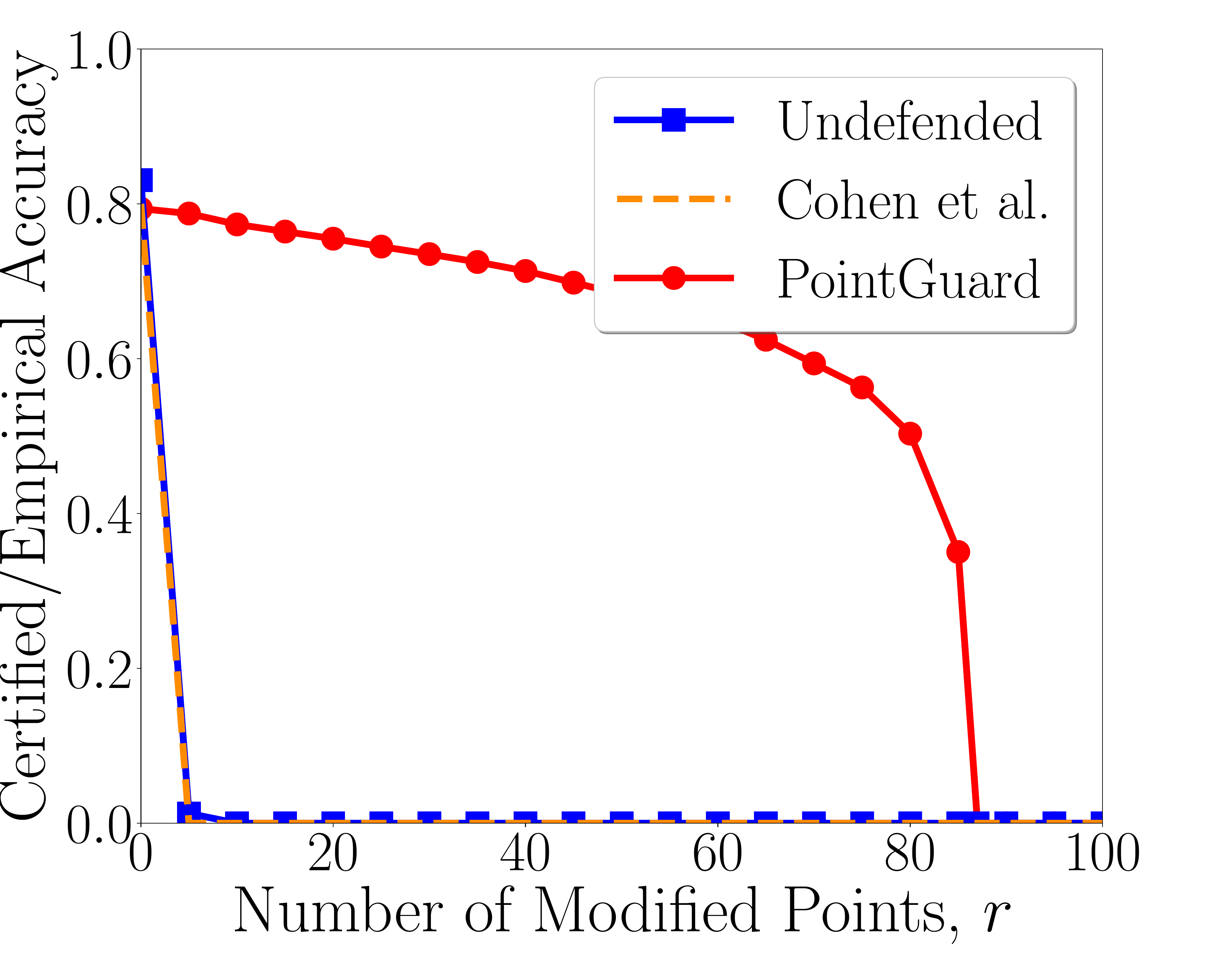}\label{compare_modification_attack} }
\subfloat[Point addition attacks.]{\includegraphics[width=0.24\textwidth]{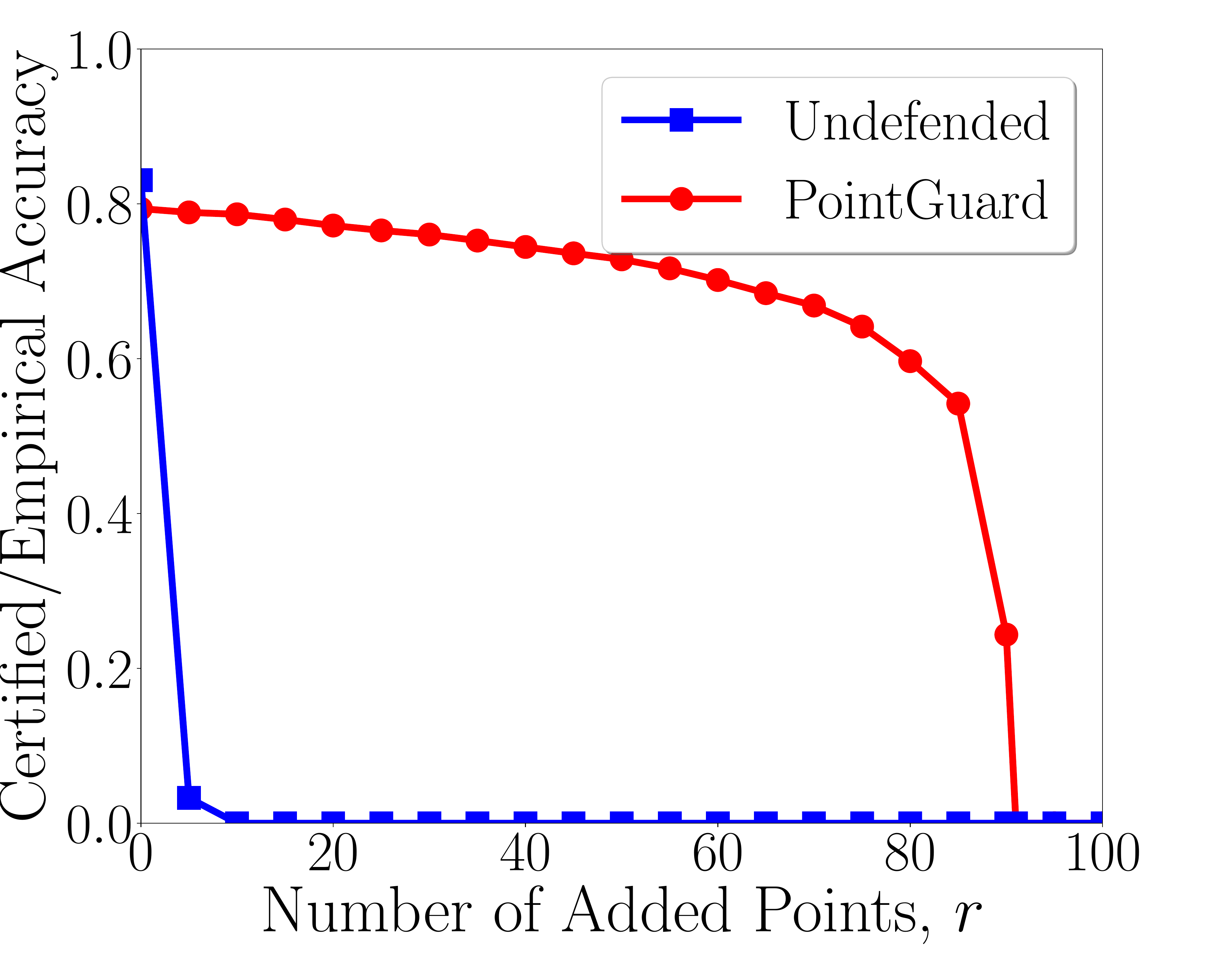}\label{compare_addition_attack}}
\subfloat[Point deletion attacks.]{\includegraphics[width=0.24\textwidth]{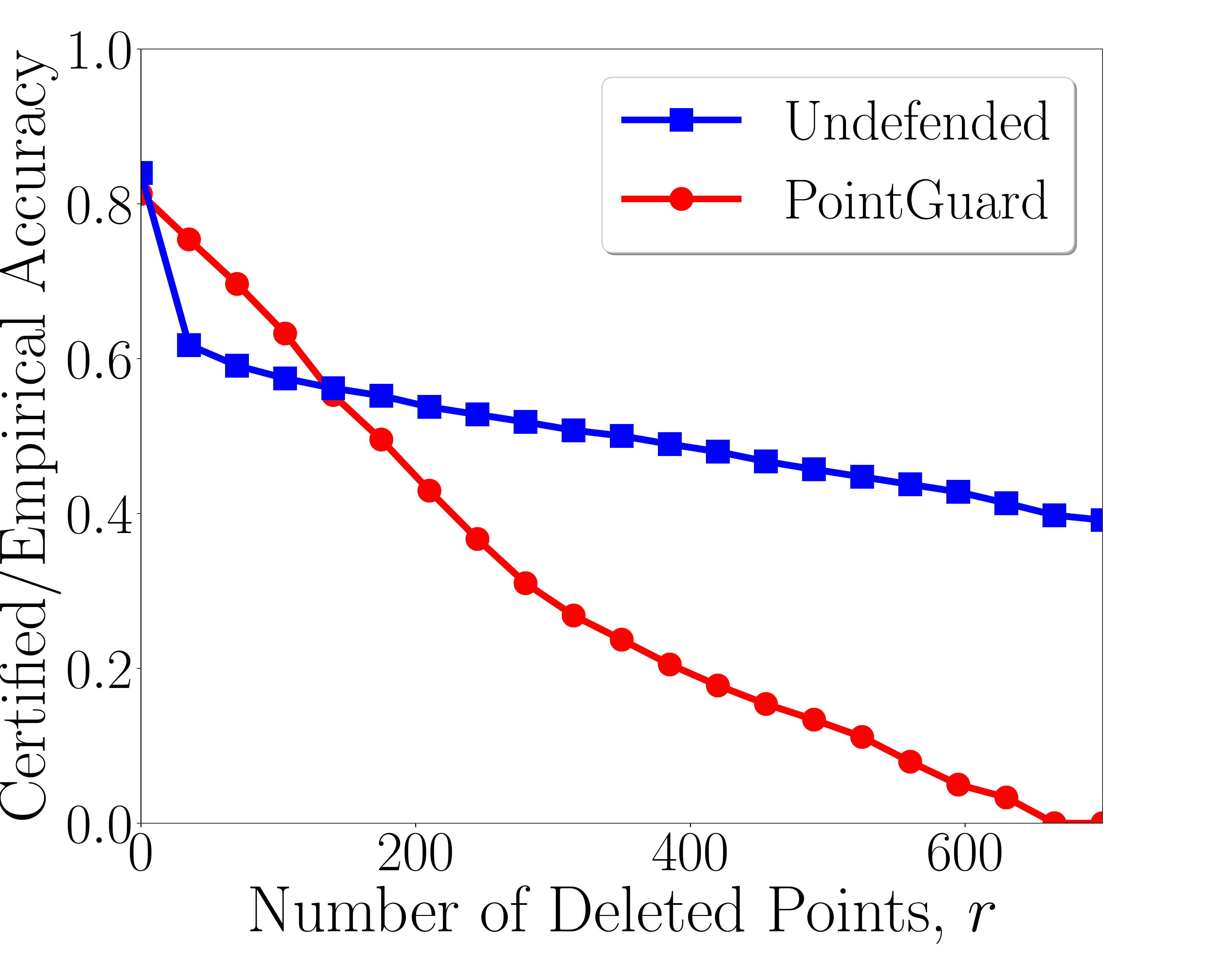}\label{compare_deletion_attack}}
\caption{Comparing different methods under different attacks on  ScanNet.}
\label{fig:scannet_1}
\end{figure*}

\begin{figure*}[!t]
\centering
\subfloat[Training with vs w/o subsampling]{\includegraphics[width=0.24\textwidth]{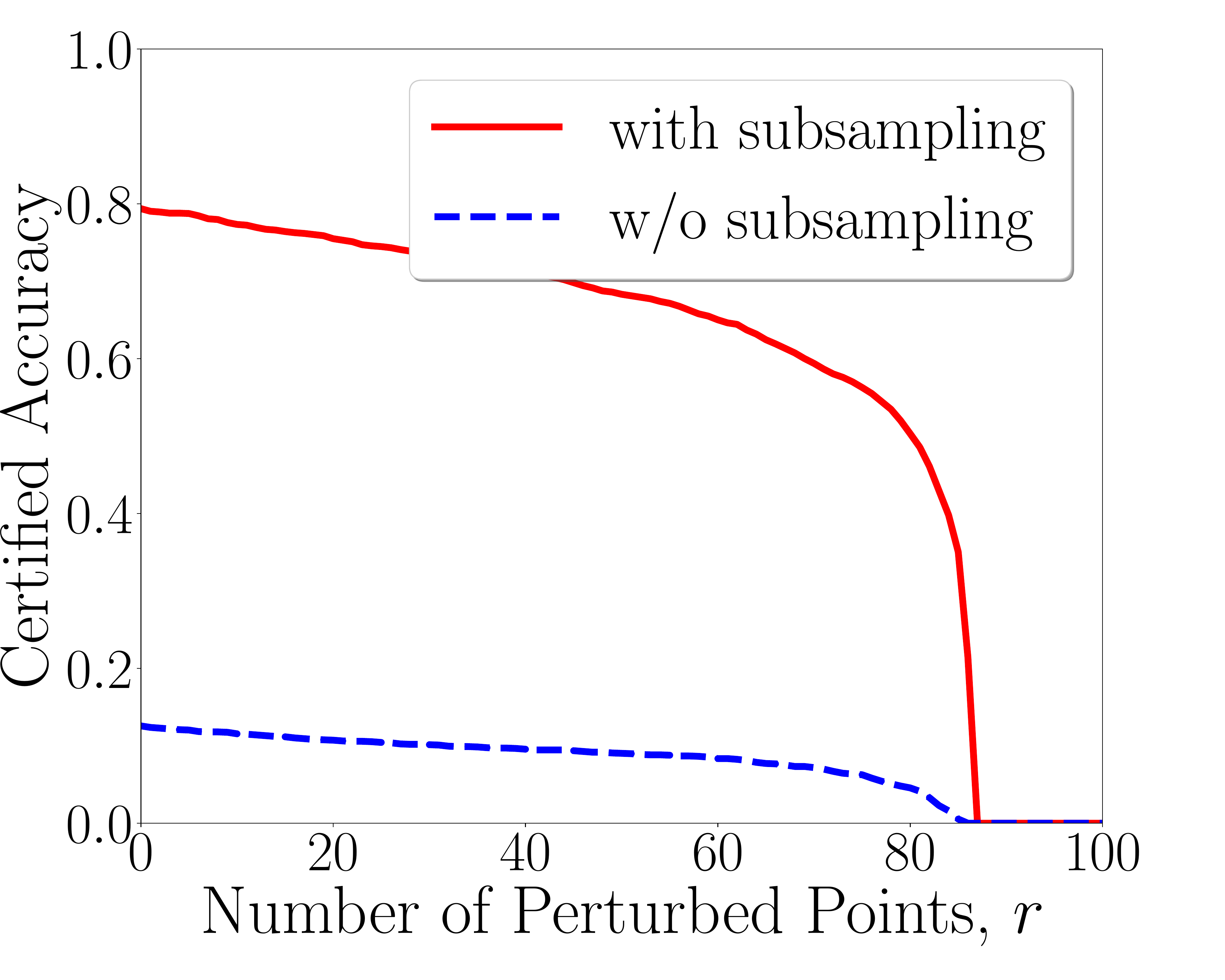}\label{impact_of_subsampling_1}}
\subfloat[Impact of $k$]{\includegraphics[width=0.24\textwidth]{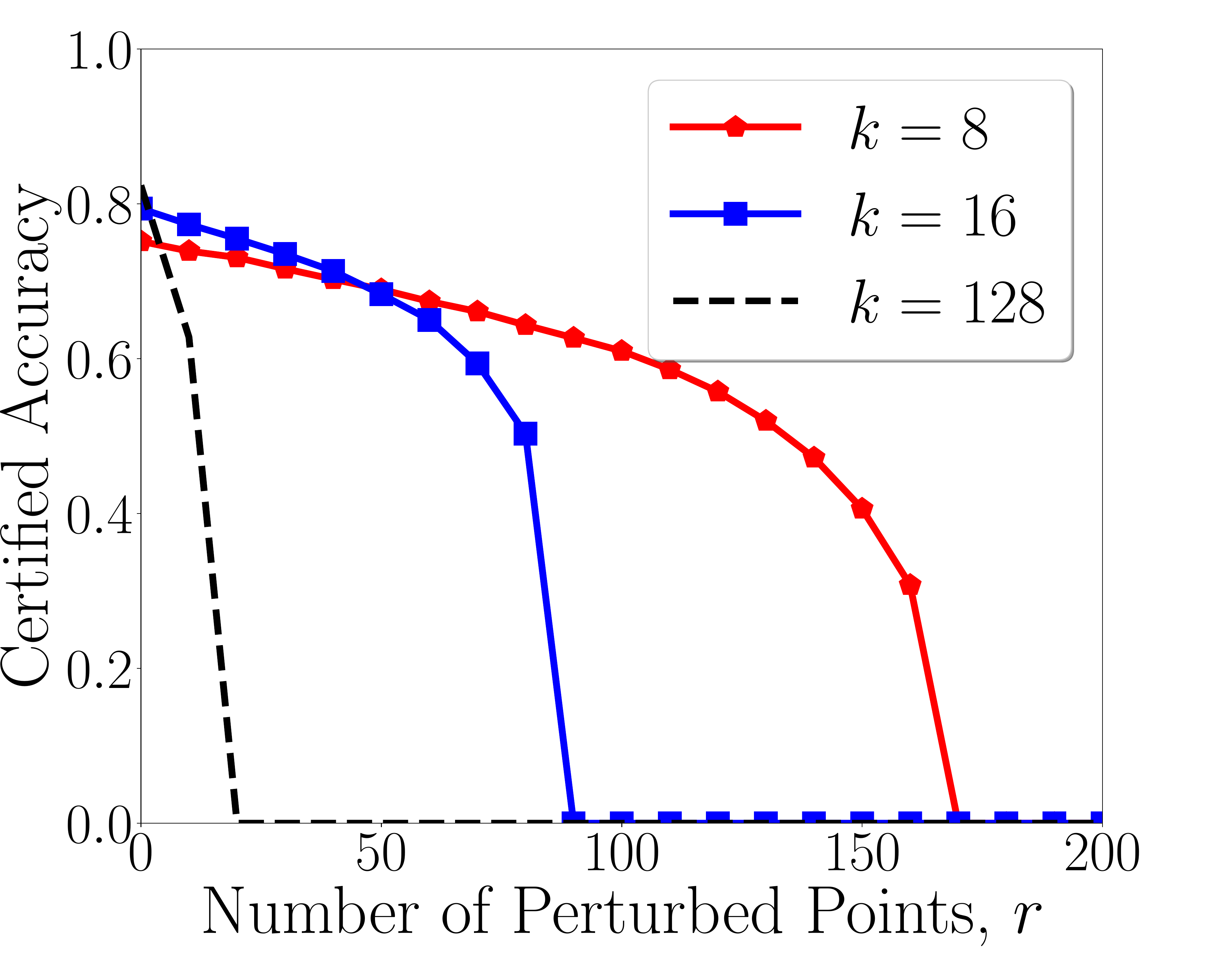}\label{impact_of_k_1}}
\subfloat[Impact of $\alpha$]{\includegraphics[width=0.24\textwidth]{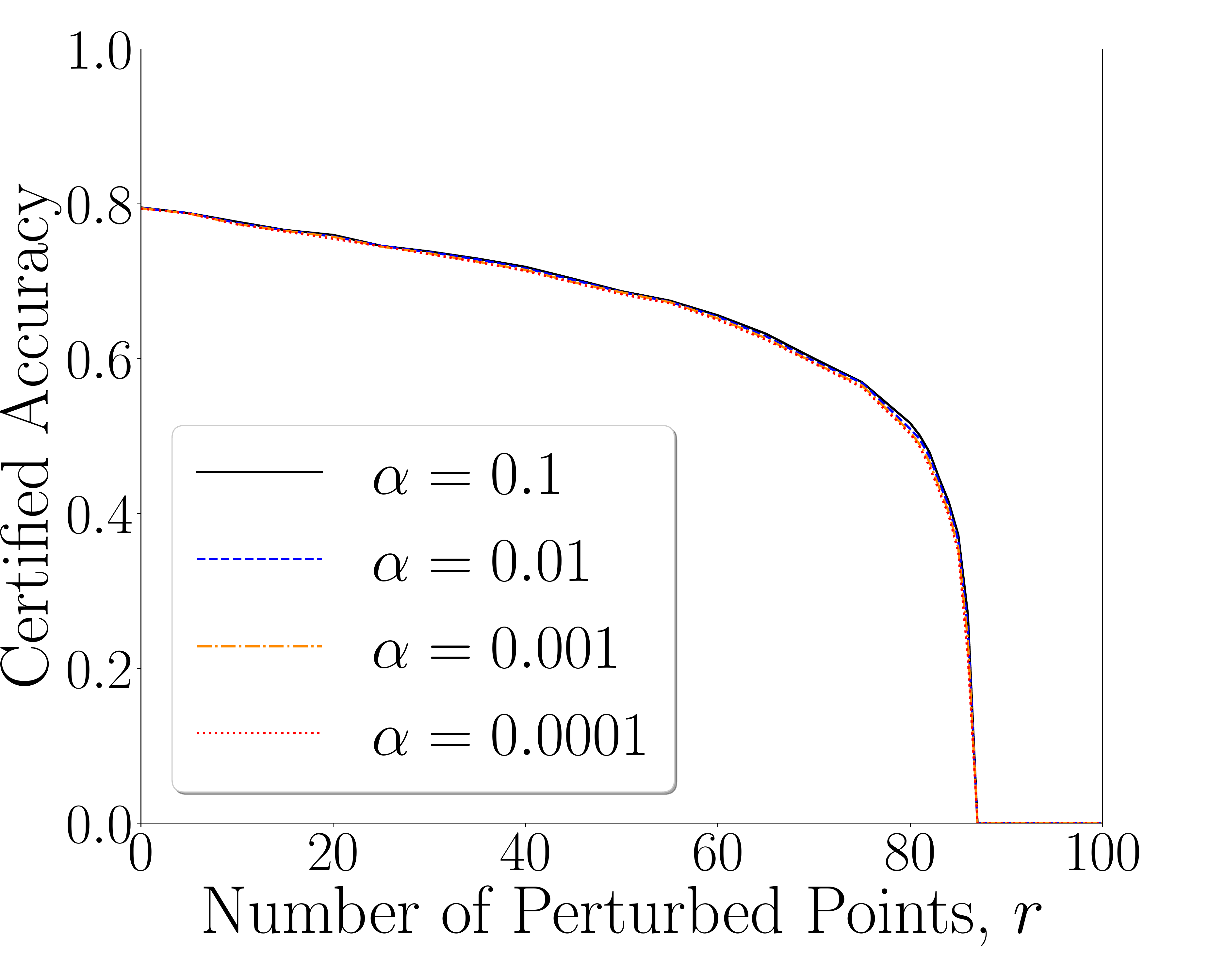}\label{impact_of_alpha_1}}
\subfloat[Impact of $N$]{\includegraphics[width=0.24\textwidth]{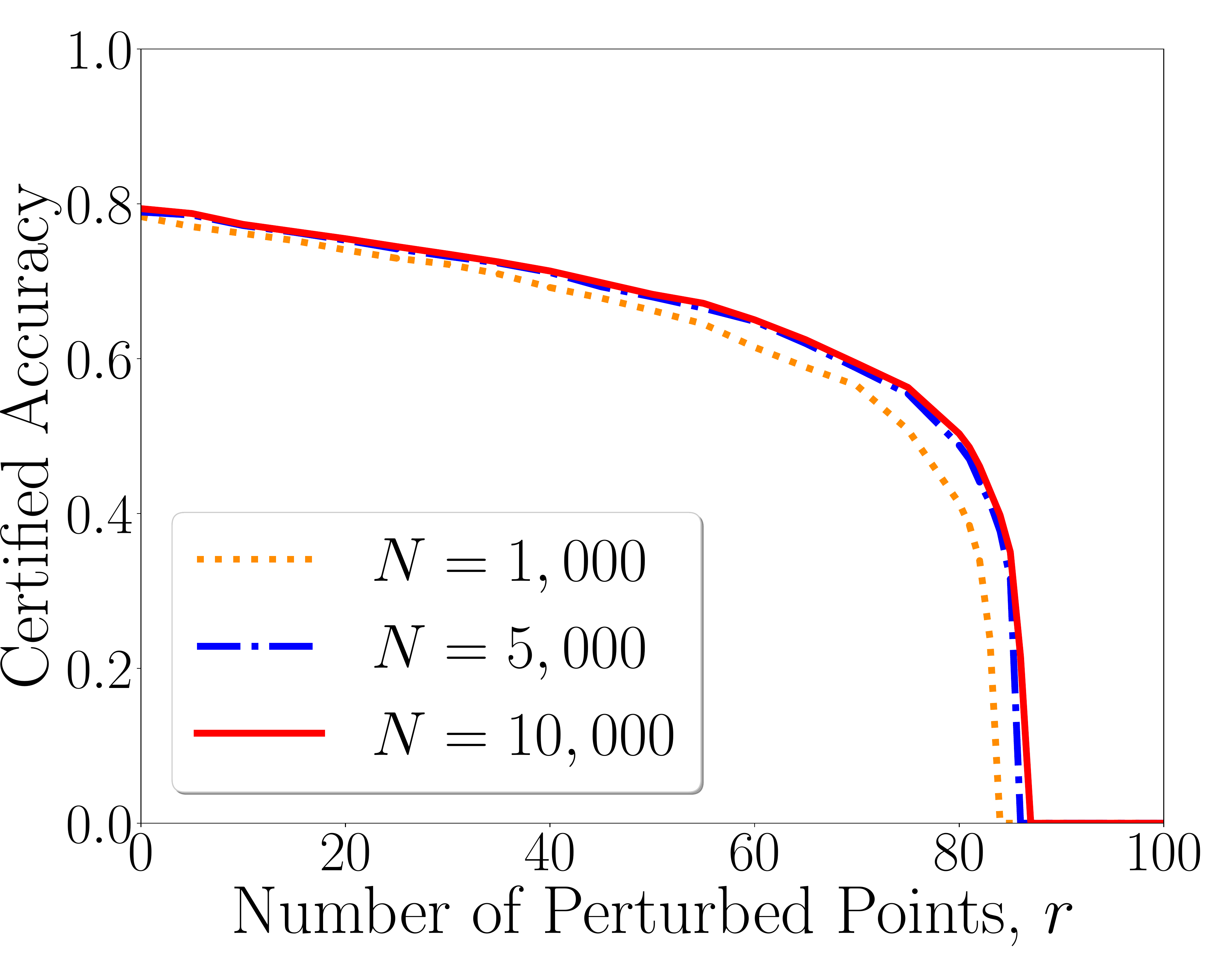}\label{impact_of_N_1}}

\caption{(a) Training the point cloud classifier with vs. without subsampling. (b), (c), and (d) show the impact of $k$, $\alpha$, and $N$, respectively. The dataset is ScanNet.}
\label{fig:scannet_2}
\end{figure*}

\end{document}